\documentclass[twocolumn]{IEEEtran}
\usepackage{amssymb,amsfonts,amsmath,amsthm}
\usepackage[dvips]{graphicx}
\usepackage{overpic}
\usepackage{subfigure}
\usepackage{color}
\usepackage{cases}
\usepackage{enumerate}
\usepackage{framed}
\definecolor{shadecolor}{rgb}{0.92,0.92,0.92}
\usepackage{listings}
\usepackage{graphicx}
\usepackage{float}
\usepackage{mathrsfs}
\usepackage[top=1.0in, bottom=1.08in, left=0.75in, right=0.75in]{geometry}
\usepackage{pgf}
\usepackage{tikz}
\usetikzlibrary{arrows, automata, positioning, calc, shapes}
\usepackage{comment}
\usepackage[numbers,sort&compress]{natbib}
\usepackage{caption}
\usepackage{hyperref}
\usepackage{bm}

\usepackage[ruled,linesnumbered]{algorithm2e}

\newtheorem{theorem}{Theorem}

\newtheorem{example}{Example}
\newtheorem{definition}{Definition}
\newtheorem{lemma}{Lemma}
\newtheorem{remark}{Remark}

\newtheorem{problem}{Problem}

\newcommand{\x}{\boldsymbol{x}}

\newcommand{\y}{\boldsymbol{y}}
\newcommand{\z}{\boldsymbol{z}}

\newcommand{\C}{\mathcal{C}}

\usepackage{comment}

\begin{document}

\title{Sequence Reconstruction for Substitution Channel: New Sufficient Conditions and Algorithms}

\author{\textbf{Chen~Wang}, \textbf{Eitan~Yaakobi}, \IEEEmembership{Senior member,~IEEE}, and \textbf{Yiwei~Zhang}, \IEEEmembership{Member,~IEEE}
\thanks{Chen Wang and Yiwei Zhang were supported in part by National Key Research and Development Program of China under Grant Nos. 2022YFA1004900 and 2021YFA1001000, in part by National Natural Science Foundation of China under Grant No. 12231014, and in part by Taishan Scholars Program. The research of Chen Wang was also supported in part at the Technion by a fellowship from the Lady Davis Foundation. The work of Eitan Yaakobi was Funded by the European Union (ERC, DNAStorage, 101045114 and EIC, DiDAX 101115134). Views and opinions expressed are however those of the authors only and do not necessarily reflect those of the European Union or the European Research Council Executive Agency. Neither the European Union nor the granting authority can be held responsible for them. Part of this paper \cite{wang2024how} has been presented in ITW2024. Corresponding author: Yiwei Zhang.}
\thanks{Chen Wang and Yiwei Zhang are with State Key Laboratory of Cryptography and Digital Economy Security, Key Laboratory of Cryptologic Technology and Information Security of Ministry of Education, School of Cyber Science and Technology, Shandong University, Qingdao, Shandong, 266237, China (e-mail: \href{mailto:cwang2021@mail.sdu.edu.cn}{cwang2021@mail.sdu.edu.cn}, \href{mailto:ywzhang@sdu.edu.cn}{ywzhang@sdu.edu.cn}).}
\thanks{Eitan~Yaakobi is with the Department of Computer Science, Technion --- Israel Institute of Technology, Haifa 3200003, Israel (e-mail: \href{mailto:yaakobi@cs.technion.ac.il}{yaakobi@cs.technion.ac.il})}
}

\maketitle

\begin{abstract}
In the \emph{sequence reconstruction problem}, a codeword $\x$ is transmitted through several identical channels where each channel produces a noisy read of $\x$, and the problem is to analyze how to uniquely reconstruct $\x$ based on these noisy reads. Levenshtein has studied the minimum number of reads which guarantees unique reconstruction of $\x$, which is one sufficient condition for unique reconstruction. In this paper, we move on to a different perspective and propose a new framework for unique reconstruction. Our new sufficient condition for unique reconstruction takes both the number of reads and the distances among the reads into consideration. We offer both theoretical analysis and corresponding efficient reconstruction algorithms for our reconstruction framework.
\end{abstract}

\IEEEpeerreviewmaketitle

\allowdisplaybreaks

\section{Introduction}
Assume a codeword $\x$ from a given code $\mathcal{C}$ is transmitted via several identical channels, where each channel produces a noisy read of $\x$. The \emph{sequence reconstruction problem}, first proposed in 2001 by Levenshtein~\cite{levenshtein2001efficient}, asks for how to uniquely reconstruct $\x$ in the worst case based on these multiple noisy reads. In particular, when the number of reads is only one, then unique reconstruction can be guaranteed if and only if $\mathcal{C}$ is chosen as an classical error-correcting code for this channel.
Recently, the sequence reconstruction problem gained renewed interest, mostly due to the development of DNA storage~\cite{church2012next,goldman2013towards,kosuri2014large,yazdi2015dna,yazdi2017portable,organick2018random,anavy2019data,choi2019high}. In current DNA storage techniques, each DNA string is amplified into a large number of copies via PCR, and then multiple noisy copies of the same string, after sequencing and clustering, could be read. Using these reads to reconstruct the original DNA string is exactly the sequence reconstruction problem.

In his seminal work~\cite{levenshtein2001efficient}, Levenshtein provided a proof of the fact that the minimum number of distinct reads which guarantees unique reconstruction is $N+1$, where $N$ is the maximum size of the intersection of any two error balls centered at two distinct codewords from $\mathcal{C}$. Here, the error ball centered at a codeword $\x$ refers to the set of all possible outputs of the channel when transmitting $\x$. The value $N$ is usually referred to as the {\it unique reconstruction threshold}. Levenshtein calculated this value for various channels with a single type of error, such as substitutions, transpositions, asymmetric errors, deletions, insertions, etc.~\cite{levenshtein2001efficient}.
Following Levenshtein, most works on sequence reconstruction are devoted to the combinatorial problem of calculating $N$, for given channels and codes $\mathcal{C}$, and then design corresponding efficient reconstruction algorithms when the number of reads exceeds the unique reconstruction threshold \cite{sala2017exact, gabrys2018sequence,horovitz2018reconstruction,konstantinova2007reconstruction,konstantinova2008reconstruction,abu2021levenshtein,abu2021list,chrisnata2020optimal,chrisnata2021correcting,cai2021coding,goyal2022sequence}.

While studying the unique reconstruction threshold is of its own combinatorial interest, the aim of this paper is to seek new sufficient (and even necessary, if possible) conditions other than Levenshtein's threshold, for unique reconstruction. There are at least three motivations. First, the threshold value $N$ is usually at least a polynomial of the codeword length $n$, and thus becomes impractically large as $n$ grows.  Second,  Levenshtein's threshold is indeed a worst-case analysis for the reconstruction problem (in the sense that we must collect enough reads for reconstruction), and it is natural to consider the following problem: When the number of reads is less than the threshold, is there any way to answer efficiently whether unique reconstruction is possible? And if so, what are the additional desired properties which could guarantee unique reconstruction? Finally, the number of reads for a given string is not a mechanism we can completely control in current DNA sequencing technology (It is possible to control how many reads in total are going to be sampled, but not how many reads we will collect for each strand~\cite{sabary2024survey}), and thus we should analyze whether unique reconstruction is possible when we read less than the threshold for some strands.

In this paper, we focus on the binary substitution channel and propose a new sufficient condition for unique reconstruction, which takes both the number of reads and the distances among the reads into consideration. Informally, our framework is based on the following observation: If a set of reads lie in the intersection of two error balls centered at two distinct codewords, then the distribution of these reads should be relatively concentrated. Our main contributions and some highlights are summarized as follows:
\begin{itemize}
    \item For two codewords $\x$ and $\y$ of length $n$ with minimum distance $d$, and the substitution channel with at most $t\geq \lceil \frac{d}{2} \rceil$ errors, we propose the concept of $D(n,m,t,d)$ to be the maximum value of the sum of the pairwise distances among $m$ distinct sequences in $B_t(\x)\cap B_t(\y)$. Detailed theoretical analysis on $D(n,m,t,d)$ is given.
    \item We propose a new sufficient condition for unique reconstruction: The summation of pairwise distances among a given set of $m$ reads exceeds the value $D(n,m,t,d)$. In other words, such a set of reads cannot lie in the intersection of two error balls centered at two distinct codewords, and thus unique reconstruction is theoretically guaranteed. Note that compared with Levenshtein's reconstruction threshold, our condition is a function of the actual reads and could work even if the number of reads is way less than Levenshtein's threshold. 
    \item We also provide reconstruction algorithms corresponding to our new condition. In the algorithms we discuss how to reconstruct the original sequence based on a set of reads triggering the condition (referred to as a triggering set), as well as how to find such a triggering set among a large set of reads. For the latter there are two interesting cases: either we may find a triggering subset among a large non-triggering set of reads, or we may transform a non-triggering set into a triggering multi-set.  
\end{itemize}


It is worth mentioning some related works. One related parallel work~\cite{papadopoulou2024on} assumes that the noisy reads are received sequentially and the decoder applies the majority decoding algorithm, and then the expected number of reads until successful reconstruction is studied. Several papers have considered the list reconstruction problem \cite{yaakobi2019uncertainty,Junnila2021On,abu2021list,Junnila2024The}, in which the decoder only needs to generate a list with a predetermined size containing the correct transmitted codeword. In particular, in ~\cite{Junnila2024The, Junnila2021On} the authors also considered the distance between the reads, and analyzed how to use two distinct reads with large distance to obtain a (short) list of candidate codewords. Meanwhile, the sequence reconstruction problem has also been considered in slightly different settings. The trace reconstruction problem is a variant which considers probabilistic channels (where each bit is erroneous with a certain probability) instead of combinatorial channels (where there is an upper bound on the number of errors)~\cite{cheng2025k,rubinstein2022average,grigorescu2022limitations,sima2021trace,batu2004reconstructing,cheraghchi2020coded}. Another related line of work studies reconstruction from multiple noisy views, in which the same underlying sequence is observed through several independent stochastic channels, and the goal is to characterize how these views jointly reduce uncertainty and improve recoverability (see \cite{hellman1970probability,kanaya1995asymptotics,mitzenmacher2006theory,land2006information,rameshwar2024information}).

The rest of the paper is organized as follows. Section \ref{sec:pre} introduces the relevant notations and the general framework of our new sufficient condition for unique reconstruction. In Section \ref{sec:main} we present the theoretical analysis for the parameters in our sufficient condition, and the corresponding reconstruction algorithm is presented in Section \ref{sec:alg}. Section \ref{Sec:Dis} contains more discussions about the difficulty to characterize equivalent conditions for unique reconstruction. Finally Section \ref{sec:concl} concludes the paper with a list of open problems.

\section{Preliminaries and Basic Results}\label{sec:pre}
For integers $m < n$, let $[m, n] = \{m, m+1, \ldots, n\}$ and denote $[1, n]$ as $[n]$ for short.
For any set $A$, let $|A|$ be the size of $A$, and for any $a\in [0, |A|]$, let $\binom{A}{a}$ be the family of all subsets of $A$ of size $a$. We use $\{\{\cdot\}\}$ to denote a multiset.
Let $\Sigma_2 = \{0, 1\}$ be the binary alphabet and $\Sigma_2^n$ be the set of sequences of length $n$ over $\Sigma_2$. For two sequences $\x$ and $\y$, let $\x\y$ be the concatenation of $\x$ and $\y$, and let $\x^n$ be the concatenation of $n$ copies of $\x$ itself. Specifically, $\x^0$ is the empty string $\boldsymbol{\epsilon}$.
For two sequences $\x = (x_1, x_2, \ldots, x_n)$ and $\y = (y_1, y_2, \ldots, y_n)$, let $d_H(\x, \y)$ be the Hamming distance between $\x$ and $\y$, i.e., $d_H(\x, \y) = |\{i\in [n]:x_i\neq y_i\}|$. Let $B_t(\x)$ be the Hamming ball of radius $t$ centered at $\x$, i.e., $B_t(\x) = \{\y\in\Sigma_2^n: d_H(\x, \y)\leq t\}$.
A code $\C$ is a subset of $\Sigma_2^n$, and the minimum Hamming distance of a code $\C$ is the minimum Hamming distance between any two distinct codewords in $\C$. For $\x = (x_1, x_2, \ldots, x_n)\in\Sigma_2^n$ and $S\subseteq[n]$, let $\x_S$ be the projection of $\x$ onto the coordinates $S$. 

The sequence reconstruction problem for substitution errors was first proposed by Levenshtein~\cite{levenshtein2001efficient}. Let $\C$ be a code in $\Sigma_2^n$ with minimum Hamming distance $d$. Assume that a codeword $\x\in\C$ is transmitted through several substitution channels, where each channel can cause at most $t$ substitution errors, i.e., the output set of the noisy reads is a subset of $B_t(\x)$. When $t<\left\lceil \frac{d}{2} \right\rceil$, by the property of error-correcting codes one can decode $\x$ from any single output in $B_t(\x)$. Thus, in the reconstruction problem, we always focus on the case $t\geq \left\lceil \frac{d}{2} \right\rceil$. Furthermore, in Levenshtein's reconstruction problem, it is assumed that the channels produce distinct outputs. In this paper, we adhere to these two assumptions unless otherwise stated. The goal is to reconstruct $\x$ by multiple distinct noisy copies. In \cite{levenshtein2001efficient}, it was shown that if the number of outputs is at least $N(n, t, d) + 1$, where
$$N\left(n, t, d\right) \triangleq \max_{\x\neq\y\in\Sigma_2^n, d_H\left(\x, \y\right) = d}\left\{|B_t\left(\x\right)\cap B_t\left(\y\right)|\right\},$$
then unique reconstruction of the original sequence is guaranteed. It has been proved by Levenshtein in \cite{levenshtein2001efficient} that
$$N\left(n, t, d\right)=\sum_{i=0}^{t-\left\lceil\frac{d}{2}\right\rceil}\binom{n-d}{i} \sum_{h=d-t+i}^{t-i}\binom{d}{h},$$
and we refer to this value as the {\it unique reconstruction threshold}. 

A common reconstruction algorithm in the sequence reconstruction problem for the substitution channel is the majority decoding algorithm, defined as follows.

\begin{definition}[The majority decoding algorithm]
Given a set of reads $\{\z_1,\dots,\z_m\}\subseteq B_t(\x)$, for each coordinate $i\in[n]$, let $n_{i,1}$ be the number of reads whose $i$-th coordinate equals 1, and let $n_{i,0}$ be the number of reads whose $i$-th coordinate equals 0. If $n_{i,1}>n_{i,0}$, then we decode $x_i$ as 1. If $n_{i,1}<n_{i,0}$, then we decode $x_i$ as 0. If $n_{i,1}=n_{i,0}$, then we decode $x_i$ as 1 or 0 arbitrarily.
\end{definition}

It can be easily checked that when $\mathcal{C}$ is the whole space $\Sigma_2^n$ (then $d = 1$), given a set of reads $\{\z_1,\dots,\z_m\}\subseteq B_t(\x)$ with $m=N(n, t, 1) + 1$, the reconstruction of $\x$ can be simply done by the majority decoding algorithm, since for each coordinate the majority decoding algorithm produces the correct value. The majority decoding algorithm and its variations have been widely used in the reconstruction of substitution errors~\cite{Junnila2024The,yaakobi2019uncertainty,abu2021levenshtein}.

Reading $N(n, t, d) + 1$ distinct noisy copies is one sufficient condition for unique reconstruction in the worst case. However, the value $N(n, t, d)$ becomes impractically large as $n$ grows. When the number of reads is less than the unique reconstruction threshold, is it possible to characterize some other sufficient (and even necessary, if possible) conditions for unique reconstruction? Say we transmit a codeword $\x \in \mathcal{C}$ and have a set of reads $\{\z_1,\dots,\z_m\}\subseteq B_t(\x)$. A straightforward way is to check whether $\bigcap_{i = 1}^{m}B_t(\z_i)\cap \C$  contains exactly one codeword. However, this method is highly inefficient as the size of each Hamming ball is a polynomial of $n$. It is desirable to have simpler conditions for unique reconstruction. To the best of our knowledge, this type of problem is rarely considered and seems rather non-trivial. In this paper, we initiate the research on this problem, by proposing a sufficient condition which takes both the number of reads and the distances among the reads into consideration. The key idea originates from the observation of the following extreme case, when three reads are already enough for unique reconstruction.

\begin{lemma}\label{lem:3reads}
    For positive integers $t$ and $n$, let $\x$ be a sequence in $\Sigma_2^n$ and $\{\z_1, \z_2, \z_3\}\subseteq B_t(\x)$ be such that $d_H(\z_1, \z_2) = d_H(\z_2, \z_3) = d_H(\z_3, \z_1) = 2t$. Then $\x$ can be uniquely determined by the majority decoding algorithm.
\end{lemma}

\begin{IEEEproof}
    Let $S_{i}$ be the coordinates where $\z_i$ and $\x$ differ, for $i\in[3]$. Since $\{\z_1, \z_2, \z_3\}\subseteq B_t(\x)$ and the Hamming distance between any two of them is $2t$, then it holds that $|S_i| = t$ for $i\in[3]$ and the three sets must be pairwise disjoint. Therefore, on each coordinate, at least two out of $\{\z_1, \z_2, \z_3\}$ agree with $\x$. Thus, using the majority decoding algorithm on $\{\z_1, \z_2, \z_3\}$ will result in $\x$.
\end{IEEEproof}

Although this special case may seem trivial, it can be easily shown that when $t$ is a fixed constant, $n$ is large enough, and the three reads are uniformly distributed in the ball $B_t(\x)$, the condition of Lemma \ref{lem:3reads} holds with high probability.

\begin{lemma}\label{lem:prob}
    Let $\x\in\Sigma_2^n$ and let $\{\z_1, \z_2, \z_3\}$ be a random set uniformly chosen from $\binom{B_t(\x)}{3}$, where $t$ is a given constant. When $n$ is large enough, we have that
    \begin{small}
    \begin{equation*}
        \mathrm{Pr}\left[d_H\left(\z_1, \z_2\right)\!=\! d_H\left(\z_2, \z_3\right)\!=\!d_H\left(\z_3, \z_1\right)\!=\!2t\right] 
        = 1-\Theta(n^{-1}).
    \end{equation*}
    \end{small}
\end{lemma}

The proof of this lemma, as well as some other probabilistic arguments of the paper, are left in the appendices. Inspired by Lemma \ref{lem:3reads}, we propose the following problem.

\begin{problem}\label{problem:equivalent}
Let $\C\subseteq\Sigma_2^n$ be a code with minimum Hamming distance $d$ and $\x\in\C$ be a codeword. For any fixed $m$ and $\{ \z_1, \ldots, \z_m \} \subseteq B_t(\x)$, where $m\leq N(n,t,d)$ and $t\geq \left\lceil \frac{d}{2} \right\rceil$, find a sufficient (and even necessary, if possible) condition for the unique reconstruction of $\x$ by $\{ \z_1, \ldots, \z_m \}$ with an efficient reconstruction algorithm.
\end{problem}

The trivial condition that solves Problem~\ref{problem:equivalent} is that the set of words $\{ \z_1, \ldots, \z_m \}$ belongs to the radius-$t$ ball of exactly one codeword. However, complexity wise this is not a feasible solution. Hence, we seek for a sufficient (and even necessary, if possible) condition that will be complexity-wise efficient, and the first thought that comes into mind is a characterization of the pairwise distance among the reads. Here we propose a sufficient condition, which is based on the summation of the pairwise distance among the reads. The following definition plays a key role in our framework.

\begin{definition}
Given integers $n,m,t,d$ where $m\geq 2$, let
$$D\left(n, m, t, d\right) \triangleq\hspace{-1ex} \max_{\substack{\x, \y\in\Sigma_2^n, d_H\left(\x,\y\right)=d,\\  \{\z_i\}_{i = 1}^m\subseteq B_t\left(\x\right)\cap B_t\left(\y\right)}} \left\{\sum_{1\leq i < j\leq m}\!\!d_H\left(\z_i, \z_j\right)\right\}$$
be the maximum value of the sum of the pairwise distance among $m$ distinct sequences in $B_t(\x)\cap B_t(\y)$, where $\x$ and $\y$ are any two sequences of distance $d$.
\end{definition}


With the help of the notation $D(n, m, t, d)$, our reconstruction condition is of the following form.

\begin{theorem}\label{thm:sufficient}
    Let $\C\subseteq\Sigma_2^n$ be a code with minimum Hamming distance $d$ and $\x\in\C$ be a codeword. For any fixed $m\geq 2$, and $m$ distinct sequences $\{\z_1, \ldots, \z_m\}\subseteq B_t(\x)$, if
    \begin{equation} \label{eqn:summation}
       \sum_{1\leq i < j\leq m}d_H\left(\z_i, \z_j\right)\geq D\left(n, m, t, d'\right) + 1
    \end{equation} for every $d' = d(\bm{c}_1, \bm{c}_2)$, for some $\bm{c}_1, \bm{c}_2\in\mathcal{C}$, then $\x$ can be uniquely reconstructed by $\{\z_1, \ldots, \z_m\}$.
\end{theorem}
\begin{IEEEproof}
    Prove by contradiction. Suppose unique reconstruction is not possible, then there is another sequence $\y\neq\x$ in $\C$, with $d_{H}(\x, \y)\geq d$, such that $\{\z_1, \ldots, \z_m\}\subseteq B_t(\x) \cap B_t(\y)$. By the definition of $D\left(n, m, t, d_{H}(\x, \y)\right)$, we have $$\sum_{1\leq i < j\leq m}d_H\left(\z_i, \z_j\right)\leq D\left(n, m, t, d_{H}\left(\x, \y\right)\right),$$ which contradicts the premise that $$\sum_{1\leq i < j\leq m}d_H\left(\z_i, \z_j\right)\geq D\left(n, m, t, d_{H}\left(\x, \y\right)\right)+1.\IEEEQEDhereeqn$$
\end{IEEEproof}

Theorem~\ref{thm:sufficient} provides a sufficient condition to Problem~\ref{problem:equivalent} and indicates that once we have $m$ distinct reads and the sum of their pairwise distance is strictly larger than $D(n,m,t,d')$ for every $d'\geq d$, then these reads cannot lie in the intersection of two error balls centered at two distinct codewords from a code with minimum Hamming distance $d$. Therefore, unique reconstruction is guaranteed. Note that for now we need to check Inequality (\ref{eqn:summation}) for every $d'\geq d$, since we are still unaware of the monotonicity of the values $\{D(n,m,t,d')\}_{d'\geq d}$ yet. After the theoretical analysis of these values in the next section, for a certain range of parameters we can modify Theorem \ref{thm:sufficient} and only need to check the inequality once, for either $d$ or $d+1$.

\section{Theoretical Analysis of the Value $D(n, m, t, d)$}\label{sec:main}

In this section, we analyze the value $D(n, m, t, d)$. Throughout the rest of the paper, we set $\x=(x_1,\dots,x_n)$ to be the transmitted codeword, $\y=(y_1,\dots,y_n)$ to be the other codeword with $d_H(\x,\y)=d$, and $\{\z_1,\ldots,\z_m\}\subseteq B_t(x)$ to be the reads where $\z_i=(z_{i,1},\ldots,z_{i,n})$ for $i\in[m]$. Let $S_i$ be the coordinates where $\z_i$ and $\x$ differ, i.e. $S_i = \{k\in[n]: z_{i, k}\neq x_k\}$. For each coordinate $k\in[n]$, let $c_k$ be the number of sequences among $\{\z_1,\dots,\z_m\}$ which are different from $\x$ at the coordinate $k$. The following Lemma \ref{lem:calculate distance}, which essentially follows a double counting argument, plays a crucial role in the subsequent analysis.

\begin{lemma}\label{lem:calculate distance}
    Following the previous notations, it holds that
$$\sum_{i<j}d_H\left(\z_i, \z_j\right) = m\sum_{k\in \left[n\right]}c_k - \sum_{k\in \left[n\right]}c_k^2.$$
\end{lemma}
\begin{IEEEproof}
    To calculate the summation $\sum_{i<j}d_H(\z_i, \z_j)$, we analyze the contribution of each coordinate to this summation. For any coordinate $k\in [n]$, since $c_k$ sequences out of $\{\z_1,\dots,\z_m\}$ differ from $\x$ and the rest $m-c_k$ sequences agree with $\x$, then this coordinate contributes $c_k(m-c_k)$ to the summation. Therefore, $$\sum_{i<j}d_H\left(\z_i, \z_j\right)=\sum_{k\in\left[n\right]} c_k\left(m-c_k\right)=m\sum_{k\in \left[n\right]}c_k - \sum_{k\in \left[n\right]}c_k^2. \IEEEQEDhereeqn$$
\end{IEEEproof}

Lemma \ref{lem:maxF} is useful when analyzing the value $D(n, m, t, d)$.

\begin{lemma}\label{lem:maxF}
    Let $F(c_1, c_2, \ldots, c_n) \triangleq m\sum_{k\in [n]}c_k - \sum_{k\in [n]}c_k^2$.
    If the summation $\sum_{k\in [n]} c_k = \lambda$ is fixed, then the maximal value of $F(c_1, \ldots, c_n)$ is obtained when $c_1, \ldots, c_n$ are either $\left\lceil\frac{\lambda}{n}\right\rceil$ or $\left\lfloor\frac{\lambda}{n}\right\rfloor$. Furthermore, consider the maximum value as a function of $\lambda$, then this function is increasing in $\lambda$ when $\lambda \leq \frac{mn}{2}$.
\end{lemma}
\begin{IEEEproof}
    The lemma follows from a fundamental inequality and we present the proof for completeness. For some $c_1, \ldots, c_n$ with fixed summation $\sum_{k\in [n]} c_k = \lambda$, if there exist ${k_1}$ and ${k_2}$ such that $c_{k_1} \geq c_{k_2} + 2$ (without loss of generality assume that ${k_1} < {k_2}$). Let $\Delta=F(c_1, \ldots, c_{k_1}, \ldots, c_{k_2}, \ldots, c_n)-F(c_1, \ldots, c_{k_1} - 1, \ldots, c_{k_2} + 1, \ldots, c_n)$, then it follows that
    \begin{align*}
        \Delta = & mc_{k_1} + mc_{k_2} - c_{k_1}^2 - c_{k_2}^2 \\
         &-m\left(c_{k_1} - 1\right) - m\left(c_{k_2} + 1\right) + \left(c_{k_1} - 1\right)^2 + \left(c_{k_2} + 1\right)^2\\
        = & 2c_{k_2} - 2c_{k_1} + 2 < 0.
    \end{align*}
    This implies that the maximal value of $F(c_1, \ldots, c_n)$ is obtained when $c_1, \ldots, c_n$ are almost equal, i.e., each entry is either $\left\lceil\frac{\lambda}{n}\right\rceil$ or $\left\lfloor\frac{\lambda}{n}\right\rfloor$. The last claim simply follows from the fact that $mc_k-c_k^2$ is increasing for any $c_k\leq \frac{m}{2}$.
\end{IEEEproof}

\subsection{Exact value of $D(n, m, t, d)$ when $n\geq m\left(t-\left\lceil\frac{d}{2}\right\rceil\right)+d$}

With Lemmas \ref{lem:calculate distance} and \ref{lem:maxF}, the value of $D(n, m, t, d)$ can be exactly determined when $n\geq m\left(t-\left\lceil\frac{d}{2}\right\rceil\right)+d$ as follows.

\begin{theorem}\label{thm:D(n, m, t, d)}
    For any positive integers $n, m, t, d$ with $t\geq \left\lceil \frac{d}{2} \right\rceil$ and $n\geq m\left(t-\left\lceil\frac{d}{2}\right\rceil\right)+d$, we have
    $$D\left(n, m, t, d\right) = d\left\lfloor\frac{m}{2}\right\rfloor\left\lceil\frac{m}{2}\right\rceil + m\left(m-1\right)\left(t-\left\lceil\frac{d}{2}\right\rceil\right).$$
\end{theorem}
\begin{IEEEproof}
    Without loss of generality, consider $\x = 0^n$, $\y = 1^d0^{n-d}$, and $\{\z_1,\ldots,\z_m\}\subseteq B_t(\x)\cap B_t(\y)$. Then each $S_i$ is indeed the support set of $\z_i$.
    Let $S'_i=S_i\cap[d]$ and $S''_i=S_i\cap[d+1,n]$, for  $i\in[m]$. For $i<j$, it holds that
    \begin{align}
    &d_H(\z_i,\z_j) \nonumber\\
    = &|\{k\in[n]:z_{i, k}\neq z_{j, k}\}|\\
    =&|\{k\in[d]:z_{i, k}\neq z_{j, k}\}|+|\{k\in[d+1,n]:z_{i, k}\neq z_{j, k}\}|\nonumber\\
    =&|S'_i| + |S'_j| - 2|S'_i\cap S'_j|+|S''_i| + |S''_j| - 2|S''_i\cap S''_j| \\
    \leq &|S'_i| + |S'_j| - 2|S'_i\cap S'_j|+|S''_i| + |S''_j| \nonumber \\
    \leq &|S'_i| + |S'_j| - 2|S'_i\cap S'_j|+t\nonumber\\
    &-\max\{|S'_i|, d - |S'_i|\}+t-\max\{|S'_j|, d - |S'_j|\}\nonumber\\
    \leq &|S'_i| + |S'_j| - 2|S'_i\cap S'_j| + 2t-2\left\lceil\frac{d}{2}\right\rceil.\nonumber
    \end{align}
    Here Equality (2) follows from standard inclusion-exclusion principle. Now we explain Inequality (3). Since $|S'_i|+|S''_i|=d(\z_i,\x)\leq t$ and $d-|S'_i|+|S''_i|=d(\z_i,\y)\leq t$, then it holds that $|S''_i|\leq t-\max\{|S'_i|, d - |S'_i|\}$. Similarly $|S''_j|\leq t-\max\{|S'_j|, d - |S'_j|\}$ and then Inequality (3) follows.

    Based on Lemma \ref{lem:calculate distance}, it holds that
    \begin{small}
    \begin{align*}
        \sum_{i<j}d_H\left(\z_i,\z_j\right)
        &\leq \sum_{i<j}\left(|S'_i| + |S'_j| - 2|S'_i\cap S'_j| + 2t-2\left\lceil\frac{d}{2}\right\rceil\right)\\
        &= m\sum_{k\in \left[d\right]}c_k - \sum_{k\in \left[d\right]}c_k^2 + \binom{m}{2}\left(2t-2\left\lceil\frac{d}{2}\right\rceil\right).
    \end{align*}
    \end{small}
    
    As $c_k$ is an integer in $[0, m]$ for each $k$, it holds that  $\left(\left\lfloor\frac{m}{2}\right\rfloor - c_k\right)\left(c_k-\left\lceil\frac{m}{2}\right\rceil\right)\leq 0$. With the fact that
    \begin{small}
    \begin{equation*}
        m\!\sum_{k\in \left[d\right]}c_k-\!\sum_{k\in \left[d\right]}c_k^2\!=\!\sum_{k\in \left[d\right]}\left(\left\lfloor\frac{m}{2}\right\rfloor\!-\!c_k\right)\!\left(c_k\!-\!\left\lceil\frac{m}{2}\right\rceil\right)\!+ d\left\lfloor\frac{m}{2}\right\rfloor\!\left\lceil\frac{m}{2}\right\rceil,
    \end{equation*}
    \end{small}
    we have
    \begin{small}
    \begin{equation}\label{eqn:D}
        \sum_{i<j}d_H\left(\z_i,\z_j\right) \leq d\left\lfloor\frac{m}{2}\right\rfloor\left\lceil\frac{m}{2}\right\rceil + m\left(m-1\right)\left(t-\left\lceil\frac{d}{2}\right\rceil\right).
    \end{equation}
    \end{small}

    Next, we show that the upper bound can be obtained. For $1\leq i\leq m$, let
    \begin{small}
    \begin{equation*}
        \z_i=
        \begin{cases}
        \hspace{-0.1ex}1^{\left\lceil\frac{d}{2}\right\rceil}\hspace{-0.1ex}0^{\left\lfloor\frac{d}{2}\right\rfloor}\hspace{-0.1ex}0^{\left(i-1\right)\left(t-\left\lceil\frac{d}{2}\right\rceil\right)}\hspace{-0.1ex}1^{t-\left\lceil\frac{d}{2}\right\rceil}\hspace{-0.1ex}0^{n-d-i\left(t-\left\lceil\frac{d}{2}\right\rceil\right)},\hspace{-1.7ex}&i\leq \left\lfloor\frac{m}{2}\right\rfloor,\\
        \hspace{-0.1ex}0^{\left\lceil\frac{d}{2}\right\rceil}\hspace{-0.1ex}1^{\left\lfloor\frac{d}{2}\right\rfloor}\hspace{-0.1ex}0^{\left(i-1\right)\left(t-\left\lceil\frac{d}{2}\right\rceil\right)}\hspace{-0.1ex}1^{t-\left\lceil\frac{d}{2}\right\rceil}\hspace{-0.1ex}0^{n-d-i\left(t-\left\lceil\frac{d}{2}\right\rceil\right)},\hspace{-1.7ex}&\text{otherwise}.
    \end{cases}
    \end{equation*}
    \end{small}

    It is straightforward to check that $d(\x,\z_i)=t$ and $d(\y,\z_i)=t-\left\lceil\frac{d}{2}\right\rceil+\left\lfloor\frac{d}{2}\right\rfloor\leq t$ when $1\leq i\leq \left\lfloor\frac{m}{2}\right\rfloor$, and similarly $d(\y,\z_i)=t$ and $d(\x,\z_i)=t-\left\lceil\frac{d}{2}\right\rceil+\left\lfloor\frac{d}{2}\right\rfloor\leq t$ when $\left\lfloor\frac{m}{2}\right\rfloor < i \leq m$. Therefore, $\{\z_1,\dots,\z_m\}\subseteq B_t(\x)\cap B_t(\y)$.
    Essentially, the key of this construction is to set $c_k\in \left\{\left\lfloor \frac{m}{2} \right\rfloor, \left\lceil \frac{m}{2} \right\rceil\right\}$ for each $k\in[d]$, and to keep $\{S''_i:i\in[m]\}$ to be pairwise disjoint. Note that the condition $n\geq m\left(t-\left\lceil\frac{d}{2}\right\rceil\right)+d$ guarantees the disjointness of $\{S''_i:i\in[m]\}$.

    Now, to calculate the summation of their pairwise distance, it is more convenient to calculate by analyzing how much each coordinate contributes to the summation. For each $k\in[d]$, almost half of the sequences have `1' and the others have `0', and thus this coordinate contributes $\left\lfloor\frac{m}{2}\right\rfloor\left\lceil\frac{m}{2}\right\rceil$ to the summation. Moreover, each coordinate $k\in \bigcup_{1\leq i\leq m} S''_i$ contributes $m-1$ to the summation, where the size of $\bigcup_{1\leq i\leq m} S''_i$ is $m\left(t - \left\lceil\frac{d}{2}\right\rceil\right)$. Thus we have $$\sum_{i<j}d_H\left(\z_i,\z_j\right)=d\left\lfloor\frac{m}{2}\right\rfloor\left\lceil\frac{m}{2}\right\rceil+m\left(m-1\right)\left(t - \left\lceil\frac{d}{2}\right\rceil\right),$$
    achieving the upper bound in Inequality (\ref{eqn:D}). To sum up, $D(n,m,t,d)=d\left\lfloor\frac{m}{2}\right\rfloor\left\lceil\frac{m}{2}\right\rceil+m(m-1)\left(t - \left\lceil\frac{d}{2}\right\rceil\right)$.
\end{IEEEproof}

\begin{example}
    Let $m = 4$, $t = 3$, and $d = 2$. We consider the sequences $\x = (0, 0,
 \ldots, 0)\in\Sigma_2^{10}$ and $\y = (1, 1, 0, \ldots, 0)\in\Sigma_2^{10}$. Consider $4$ reads $$\z_1 \hspace{-0.1ex}=\hspace{-0.1ex} \left(1, 0, 1, 1, 0, 0, 0, 0, 0, 0\right)\hspace{-0.5ex},\z_2 \hspace{-0.1ex}=\hspace{-0.1ex} \left(1, 0, 0, 0, 1, 1, 0, 0, 0, 0\right)\hspace{-0.5ex},$$ $$\z_3 \hspace{-0.1ex}=\hspace{-0.1ex} \left(0, 1, 0, 0, 0, 0, 1, 1, 0, 0\right)\hspace{-0.5ex},\z_4 \hspace{-0.1ex}=\hspace{-0.1ex} \left(0, 1, 0, 0, 0, 0, 0, 0, 1, 1\right)\hspace{-0.5ex}.$$ It holds that $\sum_{1\leq i< j\leq 4}d_H(\z_i, \z_j) = 32 = D(10, 4, 3, 2)$.
\end{example}

With Theorem \ref{thm:D(n, m, t, d)}, we are  ready to determine the value $\max_{d'\geq d}\{D(n, m, t, d')\}$ for the case $t\geq \left\lceil \frac{d}{2} \right\rceil$ and $n\geq m\left(t-\left\lceil\frac{d}{2}\right\rceil\right)+d$.

\begin{lemma}\label{lem:monotone}
    For any positive integers $n, m, t, d$ with $t\geq \left\lceil \frac{d}{2} \right\rceil$ and $n\geq m\left(t-\left\lceil\frac{d}{2}\right\rceil\right)+d$, it holds that
    \begin{itemize}
        \item $D(n, m, t, d)\geq D(n, m, t, d+2)$ where equality holds only when $m=2$;
        \item $D(n, m, t, d)> D(n, m, t, d+1)$ when $d$ is even;
        \item $D(n, m, t, d)< D(n, m, t, d+1)$ when $d$ is odd.
    \end{itemize}
\end{lemma}
\begin{IEEEproof}
    By Theorem \ref{thm:D(n, m, t, d)}, when $t\geq \left\lceil \frac{d}{2} \right\rceil$ and $n\geq m\left(t-\left\lceil\frac{d}{2}\right\rceil\right)+d$, it holds that
    {\small\begin{align*}
        D\left(n, m, t, d\right)= &d\left\lfloor\frac{m}{2}\right\rfloor\left\lceil\frac{m}{2}\right\rceil + m\left(m-1\right)\left(t-\left\lceil\frac{d}{2}\right\rceil\right),\\
        D\left(n,\! m,\! t,\! d\!+\!2\right)= &\left(d\!+\!2\right)\!\left\lfloor\frac{m}{2}\right\rfloor\!\left\lceil\frac{m}{2}\right\rceil\!+\!m\left(m\!-\!1\right)\left(t\!-\!\left\lceil\frac{d \!+ \!2}{2}\right\rceil\right).
    \end{align*}}

    Then it holds that
    \begin{align*}
        &D\left(n, m, t, d\right) - D\left(n, m, t, d+2\right)\\
        =&m^2-m - 2 \left\lfloor\frac{m}{2}\right\rfloor\left\lceil\frac{m}{2}\right\rceil \\
        =&\left(\left\lfloor\frac{m}{2}\right\rfloor + \left\lceil\frac{m}{2}\right\rceil\right)^2-\left(\left\lfloor\frac{m}{2}\right\rfloor + \left\lceil\frac{m}{2}\right\rceil\right)- 2 \left\lfloor\frac{m}{2}\right\rfloor\left\lceil\frac{m}{2}\right\rceil\\
        =&\left\lfloor\frac{m}{2}\right\rfloor^2 - \left\lfloor\frac{m}{2}\right\rfloor + \left\lceil\frac{m}{2}\right\rceil^2  - \left\lceil\frac{m}{2}\right\rceil.
    \end{align*}
    This difference is non-negative since $\left\lceil\frac{m}{2}\right\rceil$ and $\left\lfloor\frac{m}{2}\right\rfloor$ are integers, and equality holds only when $m=2$.

    Similarly, we have
    \begin{align*}
        &D\left(n, m, t, d\right) - D\left(n, m, t, d+1\right)\\
        = &m\left(m-1\right)\left(\left\lceil\frac{d + 1}{2}\right\rceil - \left\lceil\frac{d}{2}\right\rceil\right) - \left\lfloor\frac{m}{2}\right\rfloor\left\lceil\frac{m}{2}\right\rceil.
    \end{align*}
    Note that $\left\lceil\frac{d + 1}{2}\right\rceil - \left\lceil\frac{d}{2}\right\rceil = 1$ when $d$ is even and $0$ otherwise. As $m(m-1) > \left\lfloor\frac{m}{2}\right\rfloor\left\lceil\frac{m}{2}\right\rceil$, we have that $D(n, m, t, d)> D(n, m, t, d+1)$ when $d$ is even, and $D(n, m, t, d)< D(n, m, t, d+1)$ when $d$ is odd.
\end{IEEEproof}

Now we are ready to revisit the general framework of our reconstruction condition, Theorem \ref{thm:sufficient}, into the following simpler form, which means that it suffices to check only one inequality to decide whether unique reconstruction is possible.

\begin{theorem}\label{thm:main}
Let $\C\subseteq\Sigma_2^n$ be a code with minimum Hamming distance $d$ and $\x\in\C$ be a codeword. Let $t\geq \left\lceil \frac{d}{2} \right\rceil$ and $n\geq m\left(t-\left\lceil\frac{d}{2}\right\rceil\right)+d$. For any fixed $m\geq 2$, and $m$ distinct sequences $\{\z_1, \ldots, \z_m\}\subseteq B_t(\x)$, if
    $$\sum_{i<j}d_H\left(\z_i,\z_j\right) \geq D\left(n, m, t, 2\left\lceil \frac{d}{2} \right\rceil\right) + 1,$$
    then $\x$ can be uniquely reconstructed by $\{\z_1, \ldots, \z_m\}$.
\end{theorem}
\begin{IEEEproof}
    By Lemma \ref{lem:monotone}, the maximum value among $\{D(n,m,t,d'):d'\geq d\}$ is $D\left(n, m, t, 2\left\lceil \frac{d}{2} \right\rceil\right)$. Therefore, as long as $\sum_{i<j}d_H(\z_i,\z_j) \geq D\left(n, m, t, 2\left\lceil \frac{d}{2} \right\rceil\right) + 1$, then it holds that $\sum_{i<j}d_H(\z_i,\z_j) \geq D(n, m, t, d') + 1$ for all $d'\geq d$. Following Theorem \ref{thm:sufficient}, unique reconstruction is guaranteed.
\end{IEEEproof}

\begin{remark}
Consider the case with $\mathcal{C}=\Sigma^n$ (then $d=1$). For the substitution channel with $t\geq 1$ errors, Leveshtein's reconstruction threshold would require more than
$N\left(n, t, 1\right)=2\sum_{i=0}^{t-1}\binom{n-1}{i}$
distinct reads for unique reconstruction. In our framework, for any $m$ reads where $m\geq 3$ and $n\geq m(t-1)+1$, if the $m$ reads $\z_1,\ldots, \z_m$ satisfy that $$\sum_{i < j}d_H(\z_i, \z_j) \geq 2\left\lfloor\frac{m}{2}\right\rfloor\left\lceil\frac{m}{2}\right\rceil + m\left(m-1\right)\left(t - 1\right) + 1,$$ then unique reconstruction is also guaranteed.
\end{remark}

Similar as Lemma \ref{lem:prob}, in fact the unique reconstruction condition holds with high probability. The proof of the next lemma is also left in the appendix.
\begin{lemma}\label{lem:prob_m}
    Let $\C\subseteq\Sigma_2^n$ be a code with minimum Hamming distance $d$ and $\x\in\C$ be a codeword. When $m, t, d$ are given constants, $n$ is large enough, and the $m$ reads are uniformly chosen from $\binom{B_t(\x)}{m}$, with high probability we can trigger the unique reconstruction condition in Theorem \ref{thm:main} since:
    \begin{small}
    \begin{align*}
        &\mathrm{Pr}\left[\sum_{1\leq i<j\leq m}d_H\left(\z_i, \z_j\right)> D\left(n, m, t, 2\left\lceil \frac{d}{2} \right\rceil\right)\right] \\
        \geq &\mathrm{Pr}\left[\sum_{1\leq i<j\leq m}d_H\left(\z_i, \z_j\right)=m\left(m-1\right)t\right]
        \hspace{-0.5ex}=\hspace{-0.5ex} 1-\Theta\left(n^{-1}\right)\hspace{-0.5ex}.
    \end{align*}
    \end{small}
\end{lemma}

\subsection{A general upper bound of $D(n, m, t, d)$}

In the previous subsection, we determined the value of $D(n, m, t, d)$ when $n\geq m\left(t-\left\lceil\frac{d}{2}\right\rceil\right)+d$, or equivalently, when $m\leq \frac{n-d}{t-\left\lceil \frac{d}{2} \right\rceil}$. As $m$ grows larger, the upper bound in Inequality (\ref{eqn:D}) cannot be achieved anymore, since the sets $\{S''_i:1\leq i\leq m\}$ cannot be pairwise disjoint. Therefore, the upper bound in Inequality (\ref{eqn:D}) could be further reduced. In this subsection, we present a general upper bound of $D(n, m, t, d)$ for arbitrary $n$ and $m$. While the upper bound follows a similar approach as the previous subsection, analyzing whether the upper bound can be achieved turns out to be more complicated.

\begin{theorem}\label{thm:D(n, m, t, d) for smaller n}
For any positive integers $n,m,t,d$ with $t\geq \left\lceil \frac{d}{2} \right\rceil$,  it holds that
\begin{align*}
    &D\left(n, m, t, d\right) 
    \leq d\left\lfloor\frac{m}{2}\right\rfloor\left\lceil\frac{m}{2}\right\rceil \\
    &+\left(n-d\right)\left(\left\lfloor \frac{m\left(t-\left\lceil \frac{d}{2} \right\rceil\right)}{n-d} \right\rfloor^2+\left\lfloor \frac{m\left(t-\left\lceil \frac{d}{2} \right\rceil\right)}{n-d} \right\rfloor\right) \\
    &+m\left(t-\left\lceil\frac{d}{2}\right\rceil\right)\left(m-2\left\lfloor \frac{m\left(t-\left\lceil \frac{d}{2} \right\rceil\right)}{n-d} \right\rfloor-1\right).
\end{align*}

In particular, when $(n-d)\mid m\left(t-\left\lceil\frac{d}{2}\right\rceil\right)$, it holds that 
\begin{small}
\begin{equation*}
    D\left(n, m, t, d\right)
    \leq d\!\left\lfloor\!\frac{m}{2}\!\right\rfloor\!\left\lceil\!\frac{m}{2}\!\right\rceil\!+\! m\left(t\!-\!\left\lceil\frac{d}{2}\right\rceil\right)\left(m\!-\!\frac{m\left(t\!-\!\left\lceil \frac{d}{2} \right\rceil\right)}{n-d}\right).
\end{equation*}
\end{small}
\end{theorem}

\begin{IEEEproof}
    We follow the same notations as before. Without loss of generality consider $\x = 0^n$, $\y = 1^d0^{n-d}$, and $\{\z_1,\z_2,\ldots,\z_m\}\subseteq B_t(\x)\cap B_t(\y)$. Then each $S_i$ is the support set of $\z_i$. Let $S'_i=S_i\cap[d]$ and $S''_i=S_i\cap[d+1,n]$, for $i\in[m]$.
    Note that the difference with Theorem \ref{thm:D(n, m, t, d)} is that $\{S''_i:i\in [m]\}$ may not be disjoint. By Lemma \ref{lem:calculate distance} we have
    {\small\begin{align*}
    &\sum_{i<j}d_H\left(\z_i, \z_j\right) \\
    = &\sum_{i<j}\left(\left|S'_i\right|\!+\!\left|S'_j\right|\!-\!2\left|S'_i\cap S'_j\right|\right) \!+ \!\sum_{i<j}\left(\left|S''_i\right|\!+\!\left|S''_j\right|\!-\!2\left|S''_i\cap S''_j\right|\right)\\
    =  &m\sum_{k\in \left[d\right]}c_k - \sum_{k\in \left[d\right]}c_k^2 + m\sum_{k\in \left[d+1,n\right]}c_k - \sum_{k\in \left[d+1,n\right]}c_k^2.
    \end{align*}}

    On one hand, $m\sum_{k\in [d]}c_k - \sum_{k\in [d]}c_k^2$ is upper bounded by $d\left\lfloor\frac{m}{2}\right\rfloor\left\lceil\frac{m}{2}\right\rceil$, where the maximum value can be achieved if $c_k\in\{\left\lfloor\frac{m}{2}\right\rfloor,\left\lceil\frac{m}{2}\right\rceil\}$ for every $k\in[d]$.

    On the other hand, since $|S'_i|+|S''_i|=d(\z_i,\x)\leq t$ and $d-|S'_i|+|S''_i|=d(\z_i,\y)\leq t$, then it holds that $|S''_i|\leq t-\max\{|S'_i|, d - |S'_i|\} \leq t-\left\lceil \frac{d}{2} \right\rceil$, and thus $\sum_{k\in[d+1,n]}c_k \leq m(t-\left\lceil \frac{d}{2} \right\rceil)$. Now let $\lambda\triangleq \lfloor \frac{m(t-\left\lceil \frac{d}{2} \right\rceil)}{n-d}\rfloor$ and let $m(t-\left\lceil \frac{d}{2} \right\rceil)= \lambda(n-d) + R$ with $0\leq R< \lambda$. According to Lemma \ref{lem:maxF}, $m\sum_{k\in [d+1,n]}c_k - \sum_{k\in [d+1,n]}c_k^2$ will be maximized when exactly $R$ values out of $\{c_k: k\in[d+1,n]\}$ are $\lambda+1$, and the rest $n-d-R$ values are $\lambda$. Then it holds that
    \begin{small}
    \begin{align*}
     &\sum_{i<j} d_H\left(\z_i,\z_j\right) \\
    \leq &d\left\lfloor\frac{m}{2}\right\rfloor\left\lceil\frac{m}{2}\right\rceil\!+\! R\left(m\!-\!\lambda\!-\!1\right)\left(\lambda\!+\!1\right)
    \!+\!\left(n-\!d\!-\!R\right)\left(m\!-\!\lambda\right)\lambda \\
    =  &d\left\lfloor\frac{m}{2}\right\rfloor\left\lceil\frac{m}{2}\right\rceil + \left(n-d\right)\left(m-\lambda\right)\lambda + R\left(m-2\lambda-1\right) \\
    = &d\left\lfloor\frac{m}{2}\right\rfloor\left\lceil\frac{m}{2}\right\rceil + \left(n-d\right)\left(m-\lambda\right)\lambda \\
     &+\left(m\left(t-\left\lceil\frac{d}{2}\right\rceil\right)-\lambda\left(n-d\right)\right)\left(m-2\lambda-1\right) \\
    = &d\left\lfloor\frac{m}{2}\right\rfloor\left\lceil\frac{m}{2}\right\rceil\!+\! \left(n\!-\!d\right)\lambda\left(\lambda\!+\!1\right)\!+\! m\left(t\!-\!\left\lceil\frac{d}{2}\right\rceil\right)\left(m\!-\!2\lambda\!-\!1\right) \\
    = & d\left\lfloor\frac{m}{2}\right\rfloor\left\lceil\frac{m}{2}\right\rceil \!+\!\left(n\!-\!d\right)\left(\left\lfloor \frac{m\left(t-\left\lceil \frac{d}{2} \right\rceil\right)}{n-d} \right\rfloor^2\!+\!\left\lfloor \frac{m\left(t-\left\lceil \frac{d}{2} \right\rceil\right)}{n-d} \right\rfloor\right) \\
    &+m\left(t-\left\lceil\frac{d}{2}\right\rceil\right)\left(m-2\left\lfloor \frac{m\left(t-\left\lceil \frac{d}{2} \right\rceil\right)}{n-d} \right\rfloor-1\right).
    \end{align*} 
    \end{small}
    In particular, when $(n-d)|m\left(t-\left\lceil\frac{d}{2}\right\rceil\right)$, plugging $\lambda=\frac{m(t-\left\lceil \frac{d}{2} \right\rceil)}{n-d}$ and $R=0$ into above, it holds that
    \begin{align*}
        &\sum_{i<j} d_H\left(\z_i,\z_j\right) 
        \leq d\left\lfloor\frac{m}{2}\right\rfloor\left\lceil\frac{m}{2}\right\rceil + \left(n-d\right)\left(m-\lambda\right)\lambda\\
        = & d\left\lfloor\frac{m}{2}\right\rfloor\left\lceil\frac{m}{2}\right\rceil + m\left(t\!-\!\left\lceil\frac{d}{2}\right\rceil\right)\left(m\!-\!\frac{m\left(t-\left\lceil \frac{d}{2} \right\rceil\right)}{n-d}\right). 
    \end{align*}
\end{IEEEproof}

Establishing the tightness of the general upper bound in Theorem \ref{thm:D(n, m, t, d) for smaller n} requires us to, similarly as in the proof of Theorem \ref{thm:D(n, m, t, d)}, find an example consisting of $m$ distinct reads, such that \begin{itemize}
        \item For each $k\in [d]$, $c_k = \left\lceil\frac{m}{2}\right\rceil$ or $c_k=\left\lfloor \frac{m}{2}\right\rfloor$.
        \item For each $k\in [d+1, n]$, $c_k$ are almost equal and $\sum_{k=d + 1}^{n} c_k = m\left(t-\left\lceil\frac{d}{2}\right\rceil\right).$
    \end{itemize}
This construction turns out to be nontrivial in general, where the distinctness of the reads is the main difficulty. Nevertheless, we can show the tightness under an additional divisibility condition $(t-\lceil\frac{d}{2}\rceil)\mid (n-d)$. We need the famous Baranyai’s Theorem~\cite{baranyai1974factrization} from combinatorial design theory. 
In design theory, a {\it parallel class} refers to a set of $N/K$ $K$-subsets which form a partition of an $N$-set, when $K$ divides $N$. The Baranyai’s Theorem is as follows.

\begin{theorem}[Baranyai’s Theorem~\cite{baranyai1974factrization}]
    If $K$ divides $N$, the set of all $\binom{N}{K}$ $K$-subsets of an $N$-set may be partitioned into disjoint parallel classes $\mathcal{A}_i$, $i=1, 2,\ldots, \binom{N-1}{K-1}$.
\end{theorem}

By Baranyai’s Theorem we have the following result.

\begin{theorem}
    For any positive integers $n, m, t, d$ with $t\geq \left\lceil \frac{d}{2} \right\rceil$, $m\leq \binom{n-d}{t-\left\lceil\frac{d}{2}\right\rceil}$, and $(t-\left\lceil\frac{d}{2}\right\rceil)\mid (n-d)$, it holds that
    \begin{align*}
        &D\left(n, m, t, d\right) = d\left\lfloor\frac{m}{2}\right\rfloor\left\lceil\frac{m}{2}\right\rceil \\
        &+ \left(n-d\right)\left(\left\lfloor \frac{m\left(t-\left\lceil \frac{d}{2} \right\rceil\right)}{n-d} \right\rfloor^2+\left\lfloor \frac{m\left(t-\left\lceil \frac{d}{2} \right\rceil\right)}{n-d} \right\rfloor\right) \\
        &+ m\left(t-\left\lceil\frac{d}{2}\right\rceil\right)\left(m-2\left\lfloor \frac{m\left(t-\left\lceil \frac{d}{2} \right\rceil\right)}{n-d} \right\rfloor-1\right).
    \end{align*}
\end{theorem}
\begin{IEEEproof}
    Consider two sequences $\x = 0^n$, $\y = 1^d0^{n-d}$. Continuing with Theorem \ref{thm:D(n, m, t, d) for smaller n}, it suffices to find $m$ reads $\{\z_1,\z_2,\ldots,\z_m\}\subseteq B_t(\x)\cap B_t(\y)$ such that the summation of their pairwise distance achieves the upper bound.

    By Baranyai’s Theorem, since $(t-\left\lceil\frac{d}{2}\right\rceil)\mid (n-d)$, we can partition the set of all the  $(t-\left\lceil\frac{d}{2}\right\rceil)$-subsets of the set $[d+1,n]$ into disjoint parallel classes. Arbitrarily order the parallel classes and arbitrarily order the $\frac{n-d}{t-\left\lceil \frac{d}{2} \right\rceil}$ subsets within each parallel class. Let $A_{p,q}$ be the $q$-th subset in the $p$-th parallel class, where $p\in [\binom{n-d}{t-\left\lceil\frac{d}{2}\right\rceil}(t-\left\lceil\frac{d}{2}\right\rceil)/(n-d)]$ and $q\in[\frac{n-d}{t-\lceil \frac{d}{2} \rceil}]$.


    Now we construct the support sets of the reads. For each read $\z_i$, $i\in[m]$, let $i=(p_i-1)\frac{n-d}{t-\left\lceil \frac{d}{2} \right\rceil}+q_i$ and set its support set $S_i$ as
    \begin{align*}
        S_i=\begin{cases}
            \left[\left\lceil\frac{d}{2}\right\rceil\right] \cup A_{p_i,q_i},& \text{if }  i\leq \left\lfloor\frac{m}{2}\right\rfloor,\\
            \left[\left\lceil\frac{d}{2}\right\rceil+1, d\right] \cup A_{p_i,q_i}, & \text{otherwise}.
        \end{cases}
    \end{align*}

    It is possible to verify that
    \begin{itemize}
        \item $\z_i\in B_t(0^n)\cap B_t(1^d0^{n-d})$ for any $i\in [m]$,
        \item $c_k = \left\lfloor\frac{m}{2}\right\rfloor$ for $k\in \left[\left\lceil\frac{d}{2}\right\rceil\right]$ and $c_k = \left\lceil\frac{m}{2}\right\rceil$ for $k\in \left[\left\lceil\frac{d}{2}\right\rceil+1, d\right]$,
        \item $c_k = \frac{m\left(t-\left\lceil\frac{d}{2}\right\rceil\right)}{n-d}$ for $k\in [d + 1, n]$.
    \end{itemize}

    Furthermore, Baranyai’s Theorem guarantees that the values  $\{c_k:k\in[d+1,n]\}$ are almost equal, since each parallel class covers all coordinates exactly once. Thus, $c_k \in \left\{\left\lfloor \frac{m\left(t-\left\lceil\frac{d}{2}\right\rceil\right)}{n-d} \right\rfloor, \left\lceil \frac{m\left(t-\left\lceil\frac{d}{2}\right\rceil\right)}{n-d} \right\rceil\right\}$ for $k\in [d + 1, n]$.

    Given these reads, it is routine to check that the summation of their pairwise distance achieves the upper bound in Theorem \ref{thm:D(n, m, t, d) for smaller n} and thus showing its tightness.
\end{IEEEproof}

\begin{example}
    Let $m = 8$, $t = 4$, and $d = 2$. We consider the sequences $\x = (0, 0,
 \ldots, 0)\in\Sigma_2^{11}$ and $\y = (1, 1, 0, \ldots, 0)\in\Sigma_2^{11}$. Consider $8$ reads
 \begin{align*}
     \z_1 &= \left(1, 0, 1, 1, 1, 0, 0, 0, 0, 0, 0\right),\\
     \z_2 &= \left(1, 0, 0, 0, 0, 1, 1, 1, 0, 0, 0\right),\\
     \z_3 &= \left(1, 0, 0, 0, 0, 0, 0, 0, 1, 1, 1\right),\\
     \z_4 &= \left(1, 0, 1, 0, 0, 1, 0, 0, 1, 0, 0\right),\\
     \z_5 &= \left(0, 1, 0, 1, 0, 0, 1, 0, 0, 1, 0\right),\\
     \z_6 &= \left(0, 1, 0, 0, 1, 0, 0, 1, 0, 0, 1\right),\\
     \z_7 &= \left(0, 1, 1, 0, 0, 0, 1, 0, 0, 0, 1\right),\\
     \z_8 &= \left(0, 1, 0, 1, 0, 0, 0, 1, 1, 0, 0\right).
 \end{align*}
 It can be verified that $\sum_{1\leq i< j\leq 4}d_H(\z_i, \z_j) = 158 = D(11, 8, 4, 2)$. Here the parallel classes we use arise from a $(9,3,1)$-BIBD\footnote{BIBD is short for {\it balanced incomplete block designs}. A $(v,k,\lambda)$-BIBD is a family of $k$-subsets (called blocks) of a $v$-set, such that every two elements appear together in exactly $\lambda$ blocks.} with underlying set $[3,11]$: the parallel class $\{\{3,4,5\},\{6,7,8\},\{9,10,11\}\}$ corresponds to $\{\z_1,\z_2,\z_3\}$, the parallel class $\{\{3,6,9\},\{4,7,10\},\{5,8,11\}\}$ corresponds to $\{\z_4,\z_5,\z_6\}$, and finally the first two sets in the parallel class $\{\{3,7,11\},\{4,8,9\},\{5,6,10\}\}$ corresponds to $\{\z_7,\z_8\}$.
\end{example}

\section{Explicit Reconstruction Algorithms}\label{sec:alg}

In this section, we introduce the explicit reconstruction algorithms corresponding to our reconstruction condition. This section is further divided into three subsections. The first subsection introduces the reconstruction algorithm when a set of reads triggering our unique reconstruction condition, abbreviated as {\it a triggering set}, is given. The second subsection discusses how to find such a triggering set among a large number of reads. The third subsection considers the case when repeated reads are allowed and thus a triggering multi-set can also be used for reconstruction.

\subsection{The reconstruction algorithm with a triggering set}

For a code $\C\subseteq\Sigma_2^n$ with minimum Hamming distance $d$, let $\mathcal{D}_{\mathcal{C}}$ be the minimum distance decoder corresponding to $\mathcal{C}$, i.e., for any $\x\in \mathcal{C}$ and $\y\in B_{\left\lfloor\frac{d-1}{2}\right\rfloor}(\x)$, $\mathcal{D}_{\mathcal{C}}(\y) = \x$. Algorithm \ref{alg:reconstruct} works when we have a set of reads which have triggered the unique reconstruction condition presented in Theorem \ref{thm:main}.

Algorithm \ref{alg:reconstruct} first applies a majority voting method, with a pre-set parameter $\tau$, to determine each coordinate as $0$, $1$, or an undetermined value denoted by a question-mark. Here, a coordinate will be determined as a question-mark if the number of appearances of 0 and 1 at this coordinate among the reads do not differ too much, with respect to the parameter $\tau$. After Step $20$, a sequence $\z\in\{0,1,?\}^n$ is generated. Then a brute-force search is run over all possible values on the undetermined coordinates of $\z$, and the decoder of $\mathcal{C}$ is applied to get many candidates. For each candidate $\widehat{\x}$, in Step $24$ Algorithm \ref{alg:reconstruct} checks if all the reads belong to the error ball centered at $\widehat{\x}$. When the unique reconstruction condition is satisfied, we can prove that by selecting the appropriate parameter $\tau$, exactly one candidate $\widehat{\x}$ reveals and it is exactly the desired $\x$. It should be noted that the majority voting method with a pre-set parameter $\tau$ is not new. It first appeared in~\cite{yaakobi2019uncertainty} and later in other papers about sequence reconstruction such as~\cite{abu2021levenshtein}.

\vspace{-2mm}

\begin{algorithm}
    \SetKwInOut{Input}{Input}\SetKwInOut{Output}{Output}
    \caption{Reconstruction algorithm}\label{alg:reconstruct}

    \Input{$m$ sequences $Z=\left\{\boldsymbol{z}_1, \ldots, \boldsymbol{z}_{m}\right\} \subseteq B_{t}(\boldsymbol{x})$, $\tau$}
    \Output{$\widehat{\x}$, the estimation of $\x$}

    $\z = 0^n$, $F = \varnothing$\;
    \For{$k = 1$ \KwTo $n$}{
        $m_{k, 0} = m_{k, 1} = 0$\;
        \For{$i = 1$ \KwTo $m$}{
            \eIf{$z_{i, k} = 0$}{
                $m_{k, 0} = m_{k, 0} + 1$\;
            }{
                $m_{k, 1} = m_{k, 1} + 1$\;
            }
        }
        \eIf{$|m_{k, 0} - m_{k, 1}| < \tau$}{
            $z_k = ?$, $F = F\cup \{k\}$\;
        }{
            \eIf{$m_{k, 0} > m_{k, 1}$}{
                $z_k = 0$\;
            }{
                $z_k = 1$\;
            }
        }
    }
    $U = \{\boldsymbol{u}\in\{0, 1\}^n:u_k = z_k$ for all $k\notin F\}$\;
    \For{$\boldsymbol{u}\in U$}{
        $\widehat{\x} = \mathcal{D_C}(\boldsymbol{u})$\;
        \If{$Z\subseteq B_t(\widehat{\x}$)}{
            \Output{$\widehat{\x}$}
        }
    }
\end{algorithm}

\vspace{-5mm}

\begin{lemma}\label{lem:tau}
    Given $n, m, t, d$ with  $t\geq \lceil \frac{d}{2} \rceil$, $n\geq m\left(t\!-\!\left\lceil\frac{d}{2}\right\rceil\right)\!+\!d$, and a code $\C$ with minimum Hamming distance $d$, let
    \begin{align*}
       \tau = &\tau_{n, m, t, d} \\
       = &\sqrt{\frac{4}{\left\lceil\frac{d}{2}\right\rceil}\left(m\left(m\!-\!1\right)t\!-\!D\left(n, m, t, 2\left\lceil\frac{d}{2}\right\rceil\right)\right)\! +\! 1}\!-\!m\!+\!1.
    \end{align*}
    Then, for any input $\left\{\boldsymbol{z}_1, \ldots, \boldsymbol{z}_{m}\right\} \subseteq B_{t}(\boldsymbol{x})$ such that $\sum_{i<j}d_H(\z_i, \z_j) \geq D(n, m, t,2\left\lceil\frac{d}{2}\right\rceil) + 1$, at most $\lfloor \frac{d-1}{2}\rfloor$ errors exist in the determined coordinates of $\z$.
\end{lemma}
\begin{IEEEproof}
    Without loss of generality, let $\x=0^n \in \C$ be the transmitted codeword. Algorithm \ref{alg:reconstruct} generates $\z\in\{0,1,?\}^n$ and let $E_{\z} = \{k: z_k = 1\}$ be the erroneous coordinates.

    Now we prove by contradiction. Suppose there are at least $\left\lceil\frac{d}{2}\right\rceil$ errors in $\z$, i.e., $e\triangleq|E_{\z}|\geq \left\lceil\frac{d}{2}\right\rceil$. Recall that $c_k$ is the number of sequences among $\{\z_1,\dots,\z_m\}$ which are different from $\x$ at the coordinate $k$. According to the majority decoding method of the algorithm, for any $k\in E_{\z}$, we have $c_k\geq \frac{m + \tau_{n, m, t, d}}{2}$. Thus,
    \begin{equation} \label{Eqn:5}
     \sum_{k\notin E_{\z}}c_k = \sum_{k\in\left[n\right]}c_k - \sum_{k\in E_{\z}}c_k\leq mt - \frac{m + \tau_{n, m, t, d}}{2}e.
    \end{equation}

    Then the summation of the pairwise distance among the reads satisfy
	\begin{align}
		&\sum_{i<j}d_H(\z_i, \z_j) 
        = m\sum_{k\in [n]}c_k - \sum_{k\in [n]} c_k^2 \nonumber \\
		= &\left(m\sum_{k\in E_{\z}}c_k - \sum_{k\in E_{\z}} c_k^2\right) + \left(m\sum_{k\notin E_{\z}}c_k - \sum_{k\notin E_{\z}} c_k^2\right) \nonumber \\
		= &-\sum_{k\in E_{\z}}\left(c_k-\frac{m}{2}\right)^2 + \frac{m^2}{4}e + m\sum_{k\notin E_{\z}}c_k - \sum_{k\notin E_{\z}} c_k^2 \nonumber \\
		\leq &-\sum_{k\in E_{\z}}\left(c_k-\frac{m}{2}\right)^2 + \frac{m^2}{4}e + \left(m - 1\right)\sum_{k\notin E_{\z}}c_k \nonumber \\
		\leq &\frac{m^2 - \tau_{n, m, t, d}^2}{4}e + \left(m-1\right)\left(mt - \frac{m + \tau_{n, m, t, d}}{2}e\right) \label{Eqn:6} \\
		=&m\left(m-1\right)t - \frac{\left(m+\tau_{n, m, t, d}\right)\left(m+\tau_{n, m, t, d}-2\right)}{4}e \nonumber \\
		\leq &m\left(m-1\right)t - \frac{\left(m+\tau_{n, m, t, d}-1\right)^2 - 1}{4}\left\lceil\frac{d}{2}\right\rceil \nonumber \\
		= &D\left(n,m, t, 2\left\lceil\frac{d}{2}\right\rceil\right), \nonumber
	\end{align} where Inequality (\ref{Eqn:6}) arises from plugging in $c_k\geq \frac{m + \tau_{n, m, t, d}}{2}$ and Inequality (\ref{Eqn:5}), and the last step follows from plugging in the parameter $\tau$. We arrive at a contradiction to the premise that the set of reads is a triggering set.
\end{IEEEproof}

\begin{remark}\label{rem:tau}
    When $m$ is large enough, it follows that 
    \begin{align*}
       &\tau_{n, m, t, d} \\
       = &\sqrt{\hspace{-0.2ex}\frac{4}{\left\lceil\frac{d}{2}\right\rceil}\hspace{-0.2ex}\left(\hspace{-0.2ex}m\left(m\hspace{-0.2ex}-\hspace{-0.2ex}1\right)t\hspace{-0.2ex}-\hspace{-0.2ex}D\left(n, m, t, 2\hspace{-0.2ex}\left\lceil\frac{d}{2}\right\rceil\right)\right) \hspace{-0.2ex}+\hspace{-0.2ex} 1}\hspace{-0.2ex}-\hspace{-0.2ex}m\hspace{-0.2ex}+\hspace{-0.2ex}1\\
       =&\sqrt{\hspace{-0.2ex}\frac{4}{\left\lceil\frac{d}{2}\right\rceil}\hspace{-0.2ex}\left(\hspace{-0.2ex}m\left(m\hspace{-0.2ex}-\hspace{-0.2ex}1\right)\hspace{-0.2ex}\left\lceil\frac{d}{2}\right\rceil\hspace{-0.2ex}-2\hspace{-0.2ex}\left\lceil\frac{d}{2}\right\rceil\left\lceil\frac{m}{2}\right\rceil\left\lfloor\frac{m}{2}\right\rfloor\right) \hspace{-0.2ex}+ \hspace{-0.2ex}1}\hspace{-0.2ex}-\hspace{-0.2ex}m\hspace{-0.2ex}+\hspace{-0.2ex}1\\
       =&\sqrt{\frac{4}{\left\lceil\frac{d}{2}\right\rceil}\left\lceil\frac{d}{2}\right\rceil\left(m\left(m-1\right)-2\left\lceil\frac{m}{2}\right\rceil\left\lfloor\frac{m}{2}\right\rfloor\right) + 1}-m+1\\
       =&\sqrt{2m^2 + o\left(m^2\right)}-m+1\\
       =&\left(\sqrt{2} - 1\right)m+o\left(m\right).
    \end{align*}
\end{remark}

Since there is a brute-force search module in the algorithm, its complexity highly depends on the size of $F$, which is upper bounded as follows.

\begin{lemma}\label{lem:S}
    The size of $F$ is at most $\frac{2mt}{m-\tau_{n, m, t, d}}$.
\end{lemma}
\begin{IEEEproof}
    For each coordinate $k\in [n]$, $k\in F$ means there are more than  $\frac{m-\tau_{n, m, t, d}}{2}$ sequences which are different from $\x$ at coordinate $k$. As there are at most $mt$ substitution errors in total, we have $|F|\leq \frac{2mt}{m-\tau_{n, m, t, d}}$.
\end{IEEEproof}

With the two lemmas above, the following theorem holds for our reconstruction algorithm.

\begin{theorem}\label{th:opt}
    Let $n, m, t, d$ be positive integers with $t\geq \left\lceil \frac{d}{2} \right\rceil$ and $n\!\geq\! m\left(t\!-\!\left\lceil\frac{d}{2}\right\rceil\right)\!+\!d$. Let $\C$ be a code with minimum Hamming distance $d$. For $\x\in \C$ and $m$ reads $\{\z_1, \ldots, \z_m\} \subseteq B_t(\x)$ such that $\sum_{i < j}d(\z_i, \z_j) \geq D(n, m, t, 2\left\lceil\frac{d}{2}\right\rceil) + 1$, the output of Algorithm \ref{alg:reconstruct} $\widehat{\boldsymbol{x}}$ is exactly the codeword $\x$. Furthermore, when $t, d$ are constants and $n$ is sufficiently large, the time complexity of the algorithm is $\Theta(mn+T_{\C})$, where $T_{\C}$ is the time complexity of the decoder $\mathcal{D}_{\mathcal{C}}$.
\end{theorem}

\begin{IEEEproof}
   According to Lemma~\ref{lem:tau}, Algorithm \ref{alg:reconstruct} produces a sequence $\z$, and the number of errors in the determined coordinates of $\z$ compared with $\x$ is at most $\frac{d - 1}{2}$. Therefore, during the brute-force search, we have $\boldsymbol{u} \in U$ which agrees with $\x$ on all the undetermined coordinates of $\z$ and thus $d_H(\x,\boldsymbol{u})\leq \frac{d - 1}{2}$. The decoder of the code with Hamming distance $d$ will produce $\hat{\x}=\mathcal{D}_\mathcal{C}(\boldsymbol{u})$, which is exactly $\x$. Moreover, since the unique reconstruction condition is triggered, then Algorithm \ref{alg:reconstruct} screens out the candidates which do not pass the check in Step $24$, and only produces a single output $\x$.

   It is evident that the complexity of the majority voting module is $\Theta(mn)$. According to Remark~\ref{rem:tau}, we have $\tau_{n,m,t,d} = (\sqrt{2}-1)m + o(m)$. By Lemma~\ref{lem:S}, the size of the set $F$ is bounded by $|F|\leq \frac{2mt}{m-\tau_{n, m, t, d}}$. Given that $\tau_{n,m,t,d}\approx (\sqrt{2}-1)m$ and $t$ is fixed, $|F|$ remains bounded by a constant. Consequently, the sizes of both $F$ and $U$ are constant. To iterate through the sequences in $U$, one can construct an initial sequence and subsequently toggle the values at positions specified by $F$. Since $|F|$ and $|U|$ are constants, the time complexity for constructing $U$ is also $O(1)$. For each $u \in U$, the algorithm invokes the decoder $\mathcal{D}_{\mathcal{C}}$ once and computes its Hamming distance relative to all $m$ reads. Therefore, the overall time complexity of the algorithm is $\Theta(mn + T_{\mathcal{C}})$.
\end{IEEEproof}

\begin{remark}
Note that our reconstruction algorithm is valid only if $n\geq m\left(t-\left\lceil\frac{d}{2}\right\rceil\right)+d$. Thus $m$ is at most a linear function of $n$ and the time complexity is then $\Theta(n^2+T_{\mathcal{C}})$.    
\end{remark}

\begin{example}\label{ex:reconstruction algorithm}
        Let $n = 15, m = 6, t = 3$, and $d = 5$. In this case $D(n, m, t, 2\left\lceil\frac{d}{2}\right\rceil) = D(15,6,3,6) = 54$ and $\tau_{n, m, t, d} = \tau_{15,6,3,5} = 2$. Consider the code $\mathrm{BCH}(15, 7)$ with minimum Hamming distance $d = 5$, with parity-check matrix matrix $${H} \hspace{-0.5ex}=\hspace{-0.5ex} \begin{bmatrix}
        1 & 0 & 0 & 0 & 1 & 0 & 1 & 1 & 0 & 0 & 0 & 0 & 0 & 0 & 0 \\
        1 & 1 & 0 & 0 & 1 & 1 & 1 & 0 & 1 & 0 & 0 & 0 & 0 & 0 & 0 \\
        1 & 1 & 1 & 0 & 1 & 1 & 0 & 0 & 0 & 1 & 0 & 0 & 0 & 0 & 0 \\
        0 & 1 & 1 & 1 & 0 & 1 & 1 & 0 & 0 & 0 & 1 & 0 & 0 & 0 & 0 \\
        1 & 0 & 1 & 1 & 0 & 0 & 0 & 0 & 0 & 0 & 0 & 1 & 0 & 0 & 0 \\
        0 & 1 & 0 & 1 & 1 & 0 & 0 & 0 & 0 & 0 & 0 & 0 & 1 & 0 & 0 \\
        0 & 0 & 1 & 0 & 1 & 1 & 0 & 0 & 0 & 0 & 0 & 0 & 0 & 1 & 0 \\
        0 & 0 & 0 & 1 & 0 & 1 & 1 & 0 & 0 & 0 & 0 & 0 & 0 & 0 & 1 \\
        \end{bmatrix}\hspace{-1ex}.$$

        Consider the codeword $$\x=(0, 0, 0, 0, 0, 0, 0, 0, 0, 0, 0, 0, 0, 0, 0)$$ and six reads \begin{align*}
            \z_1 &= \left(1, 0, 0, 0, 0, 0, 0, 0, 1, 1, 0, 0, 0, 0, 0\right),\\
            \z_2 &= \left(1, 0, 0, 0, 0, 0, 0, 0, 0, 0, 1, 1, 0, 0, 0\right),\\
            \z_3 &= \left(0, 0, 0, 0, 0, 0, 0, 1, 1, 0, 0, 0, 1, 0, 0\right),\\
            \z_4 &= \left(0, 0, 0, 0, 0, 0, 0, 1, 1, 0, 0, 0, 0, 1, 0\right),\\
            \z_5 &= \left(0, 1, 1, 0, 0, 0, 0, 0, 1, 0, 0, 0, 0, 0, 0\right),\\
            \z_6 &= \left(1, 0, 0, 0, 0, 1, 0, 1, 0, 0, 0, 0, 0, 0, 0\right),
        \end{align*} in $B_t(\x)$, It follows that $\sum_{i < j}d_H(\z_i, \z_j) = 66 > D(15,6,3,6)$. Applying Algorithm \ref{alg:reconstruct} to the reads $\{\z_1, \ldots, \z_6\}$. Since $\tau = 2$, during the majority-vote we leave a coordinate as question-mark when the difference of the appearances of the two symbols is less than 2. Thus we produce $\z = (?, 0, 0, 0, 0, 0, 0, ?, 1, 0, 0, 0, 0, 0, 0)$ and then we move on to the brute-force module with four sequences in $U$.
        \begin{itemize}
            \item For $\boldsymbol{u} = (1, 0, 0, 0, 0, 0, 0, 1, 1, 0, 0, 0, 0, 0, 0)$, it holds that $H\boldsymbol{u} = (0, 0, 1, 0, 1, 0, 0, 0)$, which implies that the error-pattern with Hamming wight at most $\frac{d - 1}{2}=2$ is $\boldsymbol{e} = (0, 0, 0, 0, 0, 0, 0, 0, 0, 1, 0, 1, 0, 0, 0)$ (One can check that $H\boldsymbol{u} = H\boldsymbol{e}$). In this case, we have a candidate $\hat{\x} = (1, 0, 0, 0, 0, 0, 0, 1, 1, 1, 0, 1, 0, 0, 0)$. However, this candidate cannot pass the check since the last five reads are not in $B_t(\hat{\x})$. Hence $\hat{\x}$ is rejected.
            \item For the other sequences, the decoder of $\mathrm{BCH}(15, 7)$ produces 
            $\hat{\x} = (0, 0, 0, 0, 0, 0, 0, 0, 0, 0, 0, 0, 0, 0, 0).$
        \end{itemize}

         Finally the algorithm output $\hat{\x} = \x$ as desired.
\end{example}

Furthermore, if the reads have better properties, in the sense that $\sum_{i<j}d_H(\z_i, \z_j)$ is even larger, then we may discard the parameter $\tau$ and simply run a majority-vote (in the event of a tie on any coordinate, set it as $0$ or $1$ arbitrarily) to get a sequence $\z\in\{0,1\}^n$. Apply the decoder $\mathcal{D}_\mathcal{C}$ on $\z$ will directly lead to the desired $\x$. This process is formalized in the next theorem.

\begin{theorem}
For any positive integers $n, m, t, d$ with  $t\geq \left\lceil \frac{d}{2} \right\rceil$, $n\geq m\left(t-\left\lceil\frac{d}{2}\right\rceil\right)+d$, and a code $\C$ with minimum Hamming distance $d$, consider $m$ reads $\{\z_1,\z_2,\ldots,\z_m\}\subseteq B_t(\x)$ for a codeword $\x$ with
$$\sum_{i<j}d_H\left(\z_i,\z_j\right) > m\left(m-1\right)t  - \left\lceil\frac{m}{2}\right\rceil\left(\left\lceil\frac{m}{2}\right\rceil - 1\right)\left\lceil\frac{d}{2}\right\rceil.$$
Then, the majority-vote generates a sequence $\z$ which lies in the ball $B_{\lfloor\frac{d - 1}{2}\rfloor}(\x)$, and thus $\mathcal{D}_{\mathcal{C}}(\z)=\x$.
\end{theorem}

\begin{IEEEproof}
    Without loss of generality, let $\x =0^n \in \C$ be the codeword. If $\z\notin B_{\lfloor\frac{d - 1}{2}\rfloor}(\x)$, then there exists a set of coordinates $I\subseteq[n]$, such that $|I|\geq \left\lceil\frac{d}{2}\right\rceil$ and $z_{k} \neq x_k$ for any $k \in I$. According to the majority-vote method, it holds that $c_k \geq \left\lceil\frac{m}{2}\right\rceil$ for any $k\in I$, which implies that $\sum_{k\notin I}c_k\leq mt - \left\lceil\frac{m}{2}\right\rceil |I|$. It follows that
	\begin{align*}
		&\sum_{i<j}d_H\left(\z_i, \z_j\right)\\
        = &m\sum_k c_k  - \sum_k c_k^2\\
		= &m\sum_{k\in I} c_k  - \sum_{k\in I} c_k^2 + m\sum_{k\notin I} c_k  - \sum_{k\notin I} c_k^2\\
		= &-\sum_{k\in I}\left(c_k - \frac{m}{2}\right)^2 + \frac{m^2}{4}|I| + m\sum_{k\notin I} c_k  - \sum_{k\notin I} c_k^2\\
		\leq &-|I|\left(\left\lceil\frac{m}{2}\right\rceil - \frac{m}{2}\right)^2 + \frac{m^2}{4}|I| + \left(m-1\right)\sum_{k\notin I} c_k \\
		\leq & -|I|\left(\left\lceil\frac{m}{2}\right\rceil - \frac{m}{2}\right)^2 + \frac{m^2}{4}|I| + \left(m-1\right)\left(mt \hspace{-2pt}- \hspace{-2pt}\left\lceil\frac{m}{2}\right\rceil |I|\right)\\
		= &\left\lceil\frac{m}{2}\right\rceil\left\lfloor\frac{m}{2}\right\rfloor |I| + \left(m-1\right)\left(mt - \left\lceil\frac{m}{2}\right\rceil |I|\right)\\
		= &m\left(m-1\right)t  - \left\lceil\frac{m}{2}\right\rceil\left(\left\lceil\frac{m}{2}\right\rceil - 1\right) |I|\\
		\leq &m\left(m-1\right)t  - \left\lceil\frac{m}{2}\right\rceil\left(\left\lceil\frac{m}{2}\right\rceil - 1\right)\left\lceil\frac{d}{2}\right\rceil,
	\end{align*}
	which contradicts to the premise of the lemma.
\end{IEEEproof}

\begin{example}
        Let $n = 15, m = 6, t = 3$, and $d = 5$, which implies that $$m\left(m-1\right)t  - \left\lceil\frac{m}{2}\right\rceil\left(\left\lceil\frac{m}{2}\right\rceil - 1\right)\left\lceil\frac{d}{2}\right\rceil = 72.$$

        Continuing Example \ref{ex:reconstruction algorithm}, consider the codeword $\x=(0, 0, 0, 0, 0, 0, 0, 0, 0, 0, 0, 0, 0, 0, 0)$, and six reads \begin{align*}
            \z_1 &= \left(1, 0, 0, 0, 0, 0, 0, 0, 1, 1, 0, 0, 0, 0, 0\right),\\
            \z_2 &= \left(1, 0, 0, 0, 0, 0, 0, 0, 0, 0, 1, 1, 0, 0, 0\right),\\
            \z_3 &= \left(0, 0, 0, 0, 0, 1, 0, 1, 0, 0, 0, 0, 1, 0, 0\right),\\
            \z_4 &= \left(0, 0, 0, 0, 0, 0, 0, 1, 1, 0, 0, 0, 0, 1, 0\right),\\
            \z_5 &= \left(0, 1, 1, 0, 0, 0, 0, 0, 1, 0, 0, 0, 0, 0, 0\right),\\
            \z_6 &= \left(0, 0, 0, 1, 1, 0, 0, 0, 1, 0, 0, 0, 0, 0, 0\right),
        \end{align*} in $B_t(\x)$. It follows that $\sum_{i < j}d_H(\z_i, \z_j) = 74 > 72$. Applying Algorithm \ref{alg:reconstruct} to the reads $\{\z_1, \ldots, \z_6\}$ with $\tau = 0$, we produce $\z = (0, 0, 0, 0, 0, 0, 0, 0, 1, 0, 0, 0, 0, 0, 0)$ and decode $\x = (0, 0, 0, 0, 0, 0, 0, 0, 0, 0, 0, 0, 0, 0, 0)$ as desired.
\end{example}

\subsection{Finding a triggering set from a larger number of reads}
\label{subsec:more reads}
In the previous subsection we assume that a triggering set is given. When we obtain a set of $M$ reads $\{\z_1, \ldots, \z_M\} \subseteq B_t(\x)$, how do we find such a triggering set among these reads? One may suggest to simply check the summation of the pairwise distance of all the $M$ reads. However, this is not always the best strategy. We offer an example as follows.

\begin{example}\label{emp:trigger}
    Let $n = 15, m = 8, t = 4$, and $d = 5$. In this case $D(n, m, t, 2\left\lceil\frac{d}{2}\right\rceil) = D(15,8,4,6) = 152$ and $\tau_{n, m, t, d} = \tau_{15,8,4,5} = \sqrt{97}-7 \approx 2.85$. Consider the code $\C$ in Example \ref{ex:reconstruction algorithm}, the codeword $\x=(0, 0, 0, 0, 0, 0, 0, 0, 0, 0, 0, 0, 0, 0, 0)$, and eight reads 
    \begin{align*}
    \z_1 &= \left(1, 1, 1, 1, 0, 0, 0, 0, 0, 0, 0, 0, 0, 0, 0\right),\\
    \z_2 &= \left(1, 0, 0, 0, 1, 1, 1, 0, 0, 0, 0, 0, 0, 0, 0\right),\\
    \z_3 &= \left(0, 0, 0, 0, 0, 0, 0, 1, 1, 1, 1, 0, 0, 0, 0\right),\\
    \z_4 &= \left(0, 0, 0, 0, 0, 0, 0, 1, 1, 0, 0, 1, 1, 0, 0\right),\\
    \z_5 &= \left(1, 0, 0, 0, 0, 0, 0, 1, 1, 1, 0, 0, 0, 0, 0\right),\\
    \z_6 &= \left(1, 0, 0, 0, 0, 0, 0, 1, 1, 0, 1, 0, 0, 0, 0\right),\\
    \z_7 &= \left(1, 0, 0, 0, 0, 0, 0, 1, 1, 0, 0, 1, 0, 0, 0\right),\\
    \z_8 &= \left(1, 0, 0, 0, 0, 0, 0, 1, 1, 0, 0, 0, 1, 0, 0\right),
    \end{align*} in $B_t(\x)$, It follows that $\sum_{i < j}d_H(\z_i, \z_j) = 126 \leq D(15,8,4,6)$, and thus the eight reads together do not form a triggering set. In fact, if we do apply Algorithm \ref{alg:reconstruct} to the reads $\{\z_1, \ldots, \z_8\}$ with respect to the parameter $\tau\approx 2.85$, then we will produce $\z = (1, 0, 0, 0, 0, 0, 0, 1, 1, 0, 0, 0, 0, 0, 0)$ by the majority vote and decode $\hat{\x} = (1, 0, 0, 0, 0, 0, 0, 1, 1, 1, 0, 1, 0, 0, 0)$. It will then be screened out since the first read does not belongs to $B_t(\hat{\x})$. Note that the correct sequence $\x$ does not appear as a candidate. The algorithm outputs nothing and thus fails.

    On the other hand, the first four reads already constitute a triggering set since $\sum_{1\leq i < j \leq 4} d_H(\z_i, \z_j) = 42$ is larger than $D(15, 4, 4, 6) = 36$. Applying Algorithm \ref{alg:reconstruct} on the first four reads will output the correct $\x$ as desired.
\end{example}

\begin{remark}
Example~\ref{emp:trigger} shows that it is possible that while the entire set of $M$ reads is {\it not} a triggering set, a subset of these reads might be. 
\end{remark}

Motivated by this example, in this subsection we discuss how to find a subset of $m$ reads among the $M$ available reads which can be used to uniquely reconstruct $\x$. We describe the general problem in a {\it weighted graph} model as follows.

\begin{itemize}
    \item {\it The weighted clique problem:} Consider a complete graph $K_M$ where each edge is assigned a nonnegative weight. Given $D$ and $m$, with $m\leq M$, find a clique of size $m$ in which the summation of all edge-weights is larger than $D$.
    \item {\it The Hamming weighted clique problem:} Consider a set of $M$ sequences from $\Sigma_2^n$ and let them be the vertex set of a complete graph $K_M$. Assign each edge $\{\x,\y\}$ with weight $d_H(\x,\y)$. Given $m\leq M$, find a clique of size $m$ in which the summation of all edge-weights is larger than $D(n,m,t,d)$.
\end{itemize}

By the following simple reduction to the \textbf{Clique} problem of graphs, it can be shown that the  weighted clique problem is an \textbf{NP}-hard problem.

\begin{lemma}\label{lem:np-hard}
Consider a complete graph $K_M$ where each edge $\{\x, \y\}$ has weight $\omega(\x, \y)$. For any input $m$ and $D$, it is \textbf{NP}-hard to determine whether there is a clique of size $m$ in which the summation of all edge-weights is larger than $D$.
\end{lemma}

\begin{IEEEproof}
The decision version of the \textbf{Clique} problem, which is well-known to be \textbf{NP}-hard, is of the following form: Given integers $m\leq M$ and an arbitrary undirected graph $G$ with $M$ vertices, does $G$ contain a clique of size $m$?

Given an instance of the \textbf{Clique} decision problem, we can construct an instance of the weighted clique problem as follows. Build a complete graph $K_M$ with the same vertex set as $G$, and assign  the weight of each edge as:
        $$\omega\left(i,j\right) = \begin{cases}
                1, & \text{if } \{i,j\} \in E_G, \\
                0, & \text{otherwise}.
            \end{cases}$$

Now, consider the problem of finding a clique of size $m$ in $K_M$, in which the summation of all edge-weights is larger than $D\triangleq \binom{m}{2} - 1$. Note that the weight sum for any $m$-clique in $K_M$ is at most $\binom{m}{2}$, where equality holds if and only if all $\binom{m}{2}$ edges have weight 1, or equivalently, the corresponding vertices in $G$ form a clique. Therefore, $G$ has a clique of size $m$ if and only if there is an $m$-clique in $K_M$ with total edge-weights $\binom{m}{2}$. This reduction runs in polynomial time and reduces the \textbf{Clique} decision problem to the weighted clique problem. Hence, the weighted clique problem is \textbf{NP}-hard.
\end{IEEEproof}

\begin{remark}
While the reduction establishes \textbf{NP}-hardness for the weighted clique problem, we are yet not sure if there is a similar reduction from the \textbf{Clique} problem to the Hamming weighted clique problem, which is indeed what is needed to find a triggering set. Compared with the weighted clique problem, in the Hamming weighted clique problem the weighted complete graph induced by binary sequences impose additional geometric and combinatorial constraints. In particular, the following constraints hold.
\begin{itemize}
        \item The edge weights, interpreted as Hamming distances, must be positive integers.
        \item The edge weights must satisfy the distance triangle inequality. For any three vertices $\x,\y,\z$, it must hold that $d_H(\x,\y) + d_H(\x,\z) \geq d_H(\y, \z)$.
        \item The sum of edge weights in any odd-sized clique must be even, since each coordinate will contribute an even value to the summation.
        \item Each vertex is incident with at most $\binom{n}{d}$ edges with weight $d$, for any $1\leq d \leq n$.
    \end{itemize}
Could these constraints dramatically change the difficulty of the problem? We feel negative and still conjecture that the  Hamming weighted clique problem is also \textbf{NP}-hard. However we fail to find a proof by reduction yet. We leave this problem for future research.
\end{remark}

Given the analysis above, to find a triggering set among a large set of reads is a nontrivial task. We close this subsection by introducing a pruning strategy to find triggering sets.

\begin{lemma}\label{lem:average}
For any $\{\z_1, \ldots, \z_M\}$ such that $\sum_{1\leq i<j\leq M}d_H(\z_i, \z_j) > D(n,M, t, 2\left\lceil\frac{d}{2}\right\rceil)$, and for any $m\in [M]$, define $\mathrm{Sum}(m)\triangleq \sum_{1\leq i<j\leq M, i,j\neq m} d_H(\z_i, \z_j)$. Then there exists some $m$ such that $$ \mathrm{Sum}\left(m\right) > \frac{\left(M-2\right)D\left(n,M, t, 2\left\lceil\frac{d}{2}\right\rceil\right)}{M}.$$
\end{lemma}
\begin{IEEEproof}
In the summation $\sum_{m\in[M]} \mathrm{Sum}(m)$, for each pair of reads their distance $ d_H(\z_i, \z_j)$ is calculated exactly $M-2$ times. Thus
\begin{align*}
\sum_{m\in\left[M\right]} \mathrm{Sum}\left(m\right) 
= &\left(M-2\right) \sum_{1\leq i<j\leq M}d_H\left(\z_i, \z_j\right) \\
> &\left(M-2\right) D\left(n,M, t, 2\left\lceil\frac{d}{2}\right\rceil\right),
\end{align*}
and thus by the pigeonhole principle there must exist some $m$ such that $\mathrm{Sum}(m) > \frac{(M-2)D(n,M, t, 2\left\lceil\frac{d}{2}\right\rceil)}{M}$.
\end{IEEEproof}

Following the similar idea, it is routine to show that for any $m$ sequences with the summation of their pairwise distance being at least $\frac{m(m-1)}{M(M-1)}D(n,M, t, 2\left\lceil\frac{d}{2}\right\rceil)$, there exists a subset of $m-1$ sequences with the summation of their pairwise distance being at least $\frac{(m-1)(m-2)}{M(M-1)}D(n,M, t, 2\left\lceil\frac{d}{2}\right\rceil)$. By Lemma \ref{lem:average}, given a set of $M$ reads, we  proceed as follows.
\begin{itemize}
    \item If $\sum_{1\leq i < j\leq M} d_H(\mathbf{z}_i, \mathbf{z}_j) > D\left(n, M, t, 2\left\lceil \tfrac{d}{2} \right\rceil\right)$, Algorithm~\ref{alg:reconstruct} is applied to reconstruct $\x$. Otherwise, $M$ is decremented by one and the procedure continues to the subsequent steps to seek for a triggering subset.
    \item Let $S_3 = \varnothing$. Check all triples of reads $\{\z_i,\z_j,z_k\}$. If the summation of their pairwise distance is strictly larger than $D(n,3,t,2\left\lceil\frac{d}{2}\right\rceil)$, then a triggering set of size 3 has been found. Otherwise, if the summation of their pairwise distance is strictly larger than $\frac{6 D\left(n,M, t, 2\left\lceil\frac{d}{2}\right\rceil\right)}{M(M-1)}$, then let $S_3=S_3\cup \{\{i,j,k\}\}$.
    \item Recursively perform the following steps. For every $4\leq m\leq M$, set $S_m=\varnothing$. For all sets $A\in S_{m-1}$ and $i\notin A$, consider the union $A\cup\{i\}$. If the summation of the pairwise distance among the reads indexed by $A\cup\{i\}$ is strictly larger than $D(n,m,t,2\left\lceil\frac{d}{2}\right\rceil)$, then a triggering set of size $m$ has been found. Otherwise, if the summation of their pairwise distance is strictly larger than $\frac{m\left(m - 1\right)D\left(n,M, t, 2\left\lceil\frac{d}{2}\right\rceil\right)}{M(M-1)}$, then let $S_m=S_m\cup \{A\cup \{i\}\}$.
\end{itemize}

While many branches have been pruned during the search for triggering sets, Lemma \ref{lem:average} assures that if there is indeed a triggering set $A$ of size $m$, then for any $3\leq m' <m$, there is a at least one subset of $A$ in $S_{m'}$. Thus, the search process will output the triggering set as desired. Moreover, if there is no such triggering set, then at some step the set $S_{m'}$ becomes empty and the search process stops. Once a triggering set is found, Algorithm \ref{alg:reconstruct} in the previous subsection will be applied for unique reconstruction.

\begin{example}
    Let $n = 12$, $M = 6$, $t = 4$ and $d = 6$, it follows that 
    $D(n, m = 3, t, 2\lceil\frac{d}{2}\rceil) = 18$, $D(n, m = 4, t, 2\lceil\frac{d}{2}\rceil) = 36$, $D(n, m = 5, t, 2\lceil\frac{d}{2}\rceil) = 56$, and $D(n, m = 6, t, 2\lceil\frac{d}{2}\rceil) = 84$. Without loss of generality, we assume that $\x = (0, 0, 0, 0, 0, 0, 0, 0, 0, 0, 0, 0)\in\C$ is transmitted. Consider the following six reads in $B_t(\x)$: 
    \begin{align*}
        \z_1 &= \left(1, 1, 1, 1, 0, 0, 0, 0, 0, 0, 0, 0\right),\\
        \z_2 &= \left(1, 0, 0, 0, 1, 1, 1, 0, 0, 0, 0, 0\right),\\
        \z_3 &= \left(0, 1, 0, 0, 1, 0, 0, 1, 1, 0, 0, 0\right)\\
        \z_4 &= \left(0, 0, 1, 0, 0, 1, 0, 1, 0, 1, 0, 0\right),\\
        \z_5 &= \left(0, 1, 0, 0, 0, 1, 0, 0, 0, 1, 1, 0\right),\\
        \z_6 &= \left(1, 1, 0, 0, 0, 1, 0, 1, 0, 0, 0, 0\right).
    \end{align*}
    Since $\sum_{i < j} d_H(\z_i, \z_j) = 78$, Algorithm~\ref{alg:reconstruct} cannot be applied directly to the set $\{\z_1, \ldots, \z_6\}$. 
    Setting $M = 5$, we compute 
    \begin{align*}
    \frac{3(3 - 1) D\!\left(n, M, t, 2\left\lceil \frac{d}{2} \right\rceil\right)}{M(M-1)} &= 16.8,\\
    \frac{4(4 - 1) D\!\left(n, M, t, 2\left\lceil \frac{d}{2} \right\rceil\right)}{M(M-1)} &= 33.6.
    \end{align*}

    We begin with $m = 3$ and examine all $3$-element subsets of the reads. No triple of reads has total pairwise Hamming distance exceeding $18$ (thus no triggering sets of size $3$ exist), but those with total pairwise distance greater than $16.8$ are retained (since each of them could be a subset of a triggering set of size $4$ or $5$), yielding 
    \begin{align*}
        S_3 = \{ &\{\z_1, \z_2, \z_3\}, \{\z_1, \z_2, \z_4\}, \{\z_1, \z_2, \z_5\}, \{\z_1, \z_3, \z_4\},\\ 
        &\{\z_1, \z_3, \z_5\}, \{\z_2, \z_3, \z_4\}, \{\z_2, \z_3, \z_5\}\}.
    \end{align*}

    Next, we consider $m = 4$. By construction, it suffices to examine only those $4$-subsets that contain at least one member of $S_3$. 
    No $4$-subset has total pairwise distance exceeding $36$, but those with total pairwise distance greater than $33.6$ are kept, resulting in 
    $$S_4 = \left\{ \{\z_1, \z_2, \z_3, \z_4\},\ \{\z_1, \z_2, \z_3, \z_5\}\right\}.$$

    Finally, for $m = 5$, we restrict attention to $5$-subsets containing a member of $S_4$. 
    The total pairwise Hamming distance of $\{\z_1, \z_2, \z_3, \z_4, \z_5\}$ is $58 > 56$, satisfying the required threshold. After finding this triggering set, Algorithm~\ref{alg:reconstruct} can be applied to this set.
\end{example}

\subsection{Reconstruction with repeated reads}

In Levenshtein's seminal work~\cite{levenshtein2001efficient}, it is assumed that the reads from all channels are distinct. This assumption makes sense for the study of the unique reconstruction threshold. Up to this point in the paper we have also adhered to this assumption. However, looking back on all the proofs in the previous sections, we do not need this assumption at all! In other words, our general framework and algorithms also work when we are given a set of reads with repetitions. In fact, by allowing repeated reads, the only thing that could be affected (in a good way) is the upper bound in Theorem \ref{thm:D(n, m, t, d) for smaller n}, in the sense that the upper bound can be easier to achieve.

What is more unexpected is that we can artificially make repeated reads, to change a non-triggering set of reads into a triggering multi-set. Here is an example.

\begin{example}\label{emp:rep}
Let $n = 9$, $t = 2$, and $d = 2$. Consider a code $\C$ with minimum Hamming distance $d$. In this case, for $m\in\{6,7\}$ we have $D(9, 6, 2, 2) = 48$ and $D(9, 7, 2, 2) = 66$.
Without loss of generality, we assume that $\x = (0, 0, 0, 0, 0, 0, 0, 0, 0)\in\C$ is transmitted. Consider the following six reads in $B_t(\x)$: 
\begin{align*}
\z_1 &= \left(1, 1, 0, 0, 0, 0, 0, 0, 0\right),
\z_2 = \left(0, 0, 1, 1, 0, 0, 0, 0, 0\right),\\
\z_3 &= \left(0, 0, 0, 0, 0, 0, 0, 1, 1\right), \z_4 = \left(0, 0, 1, 0, 0, 1, 0, 0, 0\right),\\
\z_5 &= \left(0, 0, 1, 0, 0, 0, 1, 0, 0\right), \z_6 = \left(0, 0, 1, 0, 1, 0, 0, 0, 0\right).
\end{align*} Note that $\sum_{1\leq i < j\leq 6}d_H(\z_i, \z_j) = 48 = D(9, 6, 2, 2)$ and thus the six reads do not form a triggering set. However, by repeating $\z_1$ one more time (and thus there are 7 reads), it holds that $$2\sum_{2\leq j\leq 6}\hspace{-0.3pt}d_H\left(\z_1, \z_j\right) + \hspace{-2pt}\sum_{2\leq i < j\leq 6}\hspace{-0.3pt}d_H\left(\z_i, \z_j\right) \hspace{-2pt}=\hspace{-2pt} 68 \hspace{-2pt}>\hspace{-2pt} D\left(9, 7, 2, 2\right).$$
    Thus we can apply Algorithm \ref{alg:reconstruct} to the multiset $\{\{\z_1, \z_1, \z_2, \ldots, \z_6\}\}$ to reconstruct $\x$.
\end{example}

As the example suggests, in our unique reconstruction framework, repeated reads are allowed and we can even artificially make repeated reads to trigger unique reconstruction conditions. Our general framework can be adapted to the following version.

\begin{theorem}
    Let $\C\subseteq\Sigma_2^n$ be a code with minimum Hamming distance $d$ and $\x\in\C$ be a codeword. For any fixed $m\geq2$, and $m$ distinct reads $\{\z_1, \ldots, \z_m\}\subseteq B_t(\x)$, if there exist positive integers $\ell_1,\dots,\ell_m$ such that
  $$ \sum_{1\leq i < j \leq m} \ell_i\ell_jd_H\left(\z_i,\z_j\right) \geq D\left(n,\sum_{1\leq i \leq m} \ell_i,t,2\left\lceil\frac{d}{2}\right\rceil\right)+1,$$ then $\x$ can be uniquely reconstructed.
\end{theorem}
\begin{IEEEproof}
     It can be observed that the proofs of Theorem~\ref{thm:main} and Algorithm~\ref{alg:reconstruct} do not require the reads to be distinct. Consequently, Algorithm~\ref{alg:reconstruct} can also reconstruct the codeword from reads with repetitions, provided that the total pairwise distance among the reads exceeds the threshold specified in Theorem~\ref{thm:main}.
     Consider the reads $\{\z_1, \ldots, \z_m\}$ as a multiset
   $$\{\{\underbrace{\z_1, \ldots , \z_1}_{\ell_1}, \underbrace{\z_2, \ldots , \z_2}_{\ell_2}, \ldots, \underbrace{\z_m, \ldots , \z_m}_{\ell_m}\}\}.$$

   The total summation of the pairwise distance is now $$ \sum_{1\leq i < j \leq m} \ell_i\ell_jd_H\left(\z_i,\z_j\right) \geq D\left(n,\sum_{1\leq i \leq m} \ell_i,t,2\left\lceil\frac{d}{2}\right\rceil\right)+1,$$
   and thus it is a triggering multiset and we can apply Algorithm~\ref{alg:reconstruct} to reconstruct $\x$.
\end{IEEEproof}

Look back on Example \ref{emp:rep}. While we have interpreted it as an example that repeated reads could constitute a triggering multi-set, one may have observed that another explanation is that the first three reads already constitute a triggering set of size three. Given a set of $M$ distinct reads which is not a triggering set itself, either we may try to find a smaller triggering set as in Subsection \ref{subsec:more reads}, or we may try to artificially repeat some of the reads to find a triggering multi-set. It is not obvious which way performs better, and we conjecture that if one way works then the other way will also work.

\section{Further discussions on equivalent conditions of unique reconstruction}\label{Sec:Dis}

Levenshtein's reconstruction threshold and our framework in this paper are two sufficient conditions for unique reconstruction. It is natural to ask if we can characterize a both sufficient and necessary condition, i.e., an equivalent condition, for unique reconstruction.

Consider the case when $\mathcal{C}=\Sigma_2^n$. Suppose there are only
two reads $\{\z_1,\z_2\}\subseteq B_t(\x)$. One can find many candidates $\y$ as follows. For all the coordinates that $\z_1$ and $\z_2$ agree, let $\y$ also have the same symbol. For all the coordinates that $\z_1$ and $\z_2$ differ, let $\y$ have the same symbols as $\z_1$ on a random subset of half of the coordinates and let $\y$ have the same symbols as $\z_2$ on the other half. Then it is routine to check that $\{\z_1,\z_2\}\subseteq B_t(\y)$. Therefore,
with only two reads $\{\z_1,\z_2\}$, one can never guarantee unique reconstruction.

For $m=3$ and $\mathcal{C}=\Sigma_2^n$, recall that we have discussed this case earlier in Lemma \ref{lem:3reads} as a toy example. Next we present an equivalent condition for unique reconstruction for $m=3$ and $\mathcal{C}=\Sigma_2^n$.

\begin{theorem}\label{thm:3equiv}
For any $\x\in\Sigma_2^n$ and $\{\z_1,\z_2,\z_3\}\subseteq B_t(\x)$, $\x$ can be uniquely reconstructed from $\{\z_1,\z_2,\z_3\}$ if and only if the multiset of the pairwise distance among the three reads is $\{\{2t,2t,2t\}\}$ or $\{\{2t,2t-1,2t-1\}\}$.
\end{theorem}

\begin{IEEEproof}
Without loss of generality let $\x=0^n$. Let $S_1,S_2,S_3$ be the support sets of the three reads $\z_1,\z_2,\z_3$.

First, suppose two of the support sets have non-empty intersection, say $S_1$ and $S_2$. In this case, we can pick any coordinate $i\in S_1\cap S_2$ and $j\in S_3$. Let $\y$ be the sequence with support set $\{i,j\}$. Then it holds that $\{\z_1,\z_2,\z_3\}\subseteq B_t(\y)$ and thus unique reconstruction is not possible.

Therefore, to guarantee unique reconstruction, the three support sets must be pairwise disjoint. Next, suppose that one of the support sets, say $S_3$, is of size $|S_3|\leq t-2$. In this case, we can pick any coordinate $i\in S_1$ and $j\in S_2$. Let $\y$ be the sequence with support set $\{i,j\}$. Again it holds that $\{\z_1,\z_2,\z_3\}\subseteq B_t(\y)$ and thus unique reconstruction is not possible. Thus each support set is of size at least $t-1$.

Finally, suppose that two of the support sets, say $S_2$ and $S_3$, have size exactly $t-1$. In this case, pick any coordinate $i\in S_1$. Let $\y$ be the sequence with support set $\{i\}$. Again it holds that $\{\z_1,\z_2,\z_3\}\subseteq B_t(\y)$ and thus unique reconstruction is not possible.

To sum up, unique reconstruction is impossible unless $S_1,S_2,S_3$ are disjoint and the sizes of these three sets are either $\{\{t,t,t\}\}$ or $\{\{t,t,t-1\}\}$. From the perspective of a decoder, it means that the multiset of the pairwise distance among the three reads is $\{\{2t,2t,2t\}\}$ or $\{\{2t,2t-1,2t-1\}\}$. For these two cases a decoder can simply use a majority-vote for unique reconstruction.
\end{IEEEproof}

\begin{remark}\label{rmk:diff}
Consider $\x=0^n$ and three reads $\{\z_1,\z_2,\z_3\}\subseteq B_t(\x)$ with support sets $S_1,S_2,S_3$. Let all three sets have size exactly $t$. Suppose $|S_1\cap S_2|=1$ and $S_3$ is disjoint from the other two. The pairwise distance of three reads are $\{\{2t,2t,2t-2\}\}$ and the summation is $6t-2$. In the proof of Theorem \ref{thm:3equiv}, we have mentioned that this case cannot guarantee unique reconstruction due to the non-empty intersection of $S_1$ and $S_2$. However, the uniquely reconstructible case with pairwise distance $\{\{2t,2t-1,2t-1\}\}$ also has summation $6t-2$. This phenomenon indicates that we cannot simply have a unique reconstruction condition in the form of ``the summation of the pairwise distance of the reads exceeds a certain threshold".
\end{remark}

What about the next case $m=4$ for $\mathcal{C}=\Sigma_2^n$? We have the following facts:
\begin{itemize}
    \item Fact 1: Our reconstruction condition in Theorem \ref{thm:main} states that for four reads, a sufficient condition for unique reconstruction is that the summation of their pairwise distance is at least $D(n,4,t,2)+1=12t-3$.
    \item Fact 2: When the summation is $12t-4$, we can construct a case which does not guarantee unique reconstruction. Consider $\x=0^n$ and four reads $\{\z_1,\z_2,\z_3,\z_4\}\subseteq B_t(\x)$ with support sets $S_1,S_2,S_3,S_4$, all of size exactly $t$. Let $|S_1\cap S_2|=1$, $|S_3\cap S_4|=1$, and $S_1\cup S_2$ be disjoint from $S_3\cup S_4$. In this case the summation of the pairwise distance is $12t-4$. Pick the coordinate $i\in S_1\cap S_2$ and $j\in S_3\cap S_4$, and let $\y$ be the sequence with support set $\{i,j\}$. Then it holds that $\{\z_1,\z_2,\z_3,\z_4\}\subseteq B_t(\y)$ and thus unique reconstruction is not possible.
    \item Fact 3: A uniquely reconstructible case could have a distance summation as small as $10t$. $S_1,S_2,S_3$ are of size $t$ and mutually disjoint (and thus these three reads already guarantees unique reconstruction), and $S_4$ has a $t/3$-intersection with all the other three.
\end{itemize}

Considering these facts, it seems very difficult to characterize the equivalent condition for unique reconstruction, even for the case $m=4$. In particular, when the summation of pairwise distance among the four reads is in the range $[10t,12t-4]$, determining whether unique reconstruction is possible may require a complicated case-by-case analysis.

\section{Conclusion}\label{sec:concl}
In this paper, we prose a new framework for unique sequence reconstruction for the substitution channel. Our new sufficient condition takes both the number of reads and the summation of pairwise distance among the reads into consideration. We discuss how to find a set of reads from all available reads which can trigger our sufficient condition, and offer  an efficient corresponding reconstruction algorithm when a triggering set is given. The following problems are considered for future research:
\begin{itemize}
  \item Determine the exact value of $D(n,m,t,d)$ when $n$ is less than
  $m\left(t-\left\lceil\frac{d}{2}\right\rceil\right)+d$.
  \item Find a more efficient way to identify a triggering set from a large set of reads.
  \item Analyze more sufficient conditions for unique reconstruction. In particular, analyze equivalent conditions for unique reconstruction starting from the case $m=4$ and $\C=\Sigma_2^n$.
  \item Analyze sufficient conditions for unique reconstruction in other channels such as deletion and insertion channels.
\end{itemize}

\appendices
\section{Probability for triggering the unique reconstruction conditions}

In the appendices we present the proofs of the probabilistic arguments throughout the paper. 

\begin{IEEEproof}[Proof of Lemma \ref{lem:prob}]
    Without loss of generality, let $\x = 0^n$ and $S_i$ be the support of $\z_i$. Note that $\z_i\in B_t(\x)$,
    $$d_H(\z_1, \z_2) = d_H(\z_2, \z_3) = d_H(\z_3, \z_1) = 2t$$
    if and only if $|S_1|=|S_2|=|S_3|=t$ and $S_1, S_2, S_3$ are pairwise disjoint.

    Now we consider the possible choices of $\{\z_1, \z_2, \z_3\}$ such that $d_H(\z_1, \z_2) = d_H(\z_2, \z_3) = d_H(\z_3, \z_1) = 2t$. It can be seen that the choices of $S_1$ is $\binom{n}{t}$. As $S_2\cap S_1 = \varnothing$, the choices of $S_2$ is $\binom{n - t}{t}$. Similarly, the choices of $S_3$ is $\binom{n - 2t}{t}$. As the order of $\z_1, \z_2$, and $\z_3$ does not influence the set $\{\z_1, \z_2, \z_3\}$, the total choices of set $\{\z_1, \z_2, \z_3\}$ such that $d_H(\z_1, \z_2) = d_H(\z_2, \z_3) = d_H(\z_3, \z_1) = 2t$ is $\frac{\binom{n}{t}\binom{n-t}{t}\binom{n-2t}{t}}{3!}$. As $|B_t(\x)| = \sum_{i = 0}^t\binom{n}{i}$, it follows that
    {\small\begin{align*}
        \mathrm{Pr}\left[d_H\left(\z_1, \!\z_2\right)\! =\! d_H\!\left(\z_2,\! \z_3\right)\! =\! d_H\left(\z_3, \!\z_1\right) \!=\! 2t\right] \!= \!\frac{\binom{n}{t}\!\binom{n-t}{t}\!\binom{n-2t}{t}}{3! \binom{\sum_{i = 0}^t\binom{n}{i}}{3}}.
    \end{align*}}

    Now we estimate the probability. It can be seen that \begin{align*}
        \binom{n}{t} &= \frac{n(n - 1)\ldots(n - t + 1)}{t!} \\
        &= \frac{n^t - \sum_{i=0}^{t - 1}in^{t - 1} + \Theta(n^{t - 2})}{t!} \\
        &= \frac{n^t - \frac{t(t - 1)}{2}n^{t - 1} + \Theta(n^{t - 2})}{t!} \\
        &= \frac{n^t}{t!} - \frac{\left(t-1\right)n^{t-1}}{2\left(t-1\right)!}+\Theta\left(n^{t-2}\right).
    \end{align*}

    Similarly, we have:
    \begin{align*}
        \binom{n-t}{t} &= \frac{n^t}{t!} - \frac{\left(3t-1\right)n^{t-1}}{2\left(t-1\right)!}+\Theta\left(n^{t-2}\right),\\
        \binom{n-2t}{t} &= \frac{n^t}{t!} - \frac{\left(5t-1\right)n^{t-1}}{2\left(t-1\right)!}+\Theta\left(n^{t-2}\right),
    \end{align*}
    and
    $$\sum_{i = 0}^t\binom{n}{i} = \frac{n^t}{t!} - \frac{\left(t-3\right)n^{t-1}}{2\left(t-1\right)!}+\Theta\left(n^{t-2}\right).$$

    Therefore, it follows that 
    \begin{align*}
        &\binom{n}{t}\binom{n-t}{t}\binom{n-2t}{t} \\= &\left(\frac{n^t}{t!} - \frac{\left(t-1\right)n^{t-1}}{2\left(t-1\right)!}+\Theta\left(n^{t-2}\right)\right)\\
        &\cdot\left(\frac{n^t}{t!} - \frac{\left(3t-1\right)n^{t-1}}{2\left(t-1\right)!}+\Theta\left(n^{t-2}\right)\right)\\
        &\cdot\left(\frac{n^t}{t!} - \frac{\left(5t-1\right)n^{t-1}}{2\left(t-1\right)!}+\Theta\left(n^{t-2}\right)\right)\\
        =&\frac{n^{3t}}{\left(t!\right)^3} - \frac{\left(9t-3\right)n^{3t-1}}{2\left(t!\right)^2\left(t-1\right)!} + \Theta\left(n^{3t-2}\right),
    \end{align*}
    and 
    \begin{align*}
        3!\binom{\sum_{i = 0}^t\binom{n}{i}}{3} &= 3!\binom{\frac{n^t}{t!} - \frac{\left(t-3\right)n^{t-1}}{2\left(t-1\right)!}+\Theta\left(n^{t-2}\right)}{3} \\
        &= \left(\frac{n^t}{t!} - \frac{\left(t-3\right)n^{t-1}}{2\left(t-1\right)!}+\Theta\left(n^{t-2}\right)\right)^3 \\
        &= \frac{n^{3t}}{\left(t!\right)^3} - \frac{3\left(t-3\right)n^{3t-1}}{2\left(t!\right)^2\left(t-1\right)!} + \Theta\left(n^{3t-2}\right).
    \end{align*}

    In total, it holds that
    \begin{align*}
        &\mathrm{Pr}\left[d_H\left(\z_1, \z_2\right) = d_H\left(\z_2, \z_3\right) = d_H\left(\z_3, \z_1\right) = 2t\right]\\
        =&\frac{\frac{n^{3t}}{\left(t!\right)^3} - \frac{\left(9t-3\right)n^{3t-1}}{2\left(t!\right)^2\left(t-1\right)!} + \Theta\left(n^{3t-2}\right)}{\frac{n^{3t}}{\left(t!\right)^3} - \frac{3\left(t-3\right)n^{3t-1}}{2\left(t!\right)^2\left(t-1\right)!} + \Theta\left(n^{3t-2}\right)}\\
        =&1-\frac{\frac{\left(6t+6\right)n^{3t-1}}{2\left(t!\right)^2\left(t-1\right)!} + \Theta\left(n^{3t-2}\right)}{\frac{n^{3t}}{\left(t!\right)^3} - \frac{3\left(t-3\right)n^{3t-1}}{2\left(t!\right)^2\left(t-1\right)!} + \Theta\left(n^{3t-2}\right)}\\
        =& 1-\frac{6t\left(t+1\right)+o\left(1\right)}{2n-3t\left(t-3\right)+o\left(1\right)},
    \end{align*}
    and thus the probability to trigger the the unique reconstruction condition of Lemma \ref{lem:3reads} is $1 -\Theta(n^{-1})$.
\end{IEEEproof}

\begin{IEEEproof}[Proof of Lemma \ref{lem:prob_m}]
    It is easy to verify that $D(n, m, t, 2\left\lceil \frac{d}{2} \right\rceil) + 1\leq m(m - 1)t$. Without loss of generality, we assume that $\x = (0, 0, \ldots, 0)$. It can be calculated that
        \begin{align*}
		&\mathrm{Pr}\left[\sum_{i<j}d_H\left(\z_i, \z_j\right) = m\left(m-1\right)t\right] \\
        = &\frac{\frac{1}{m!}\prod_{i = 0}^{m - 1}\binom{n - it}{t}}{\binom{\binom{n}{t}}{m}} \\
		\geq &  \frac{\frac{1}{m!}\prod_{i = 0}^{m - 1}\binom{n - it}{t}}{\frac{1}{m!}\binom{n}{t}^m}\\
		= & \frac{\prod_{i = 0}^{m - 1}\binom{n - it}{t}}{\binom{n}{t}^m}\\
		\geq & \prod_{i = 0}^{m - 1}\prod_{j = 0}^{t - 1}\frac{n - it - j}{n}\\
        = & \frac{n^{m} - \frac{mt(mt-1)}{2}n^{m-1} + o(n^{m - 1})}{n^m}\\
        = & 1 - \Theta\left(n^{-1}\right).
	\end{align*} 
    where the first inequality follows from $\frac{n^t}{t!}\geq \binom{n}{t}$.
\end{IEEEproof}

\bibliographystyle{IEEEtran}
\bibliography{ref1}

\begin{IEEEbiographynophoto}{Chen Wang} received the B.A. and M.S. degrees in Mathematics from the University of Science and Technology of China, Hefei, Anhui, China, in 2019 and 2021, respectively, and the Ph.D. degree from Shandong University, Qingdao, Shandong, China, in 2025. He is currently with the Department of Computer Science, Technion --- Israel Institute of Technology. His research interests include combinatorics, DNA storage, and private information retrieval.
\end{IEEEbiographynophoto} 

\begin{IEEEbiographynophoto}{Eitan Yaakobi} (S'07--M'12--SM'17) is a Professor at the Computer Science Department at the Technion --- Israel Institute of Technology. He also holds a courtesy appointment in the Technion's Electrical and Computer Engineering (ECE) Department. He received the B.A. degrees in computer science and mathematics, and the M.Sc. degree in computer science from the Technion --- Israel Institute of Technology, Haifa, Israel, in 2005 and 2007, respectively, and the Ph.D. degree in electrical engineering from the University of California, San Diego, in 2011. Between 2011-2013, he was a postdoctoral researcher in the department of Electrical Engineering at the California Institute of Technology and at the Center for Memory and Recording Research at the University of California, San Diego. His research interests include information and coding theory with applications to non-volatile memories, associative memories, DNA storage, data storage and retrieval, and private information retrieval. He received the Marconi Society Young Scholar in 2009 and the Intel Ph.D. Fellowship in 2010-2011. Between 2020 and 2023, he served as an Associate Editor for Coding and Decoding for the \textsc{IEEE Transactions on Information Theory} and since October 2024 he serves as an Associate Editor for the \textsc{IEEE Transactions on Molecular, Biological, and Multi-Scale Communications}. Since 2016, he is affiliated with the Center for Memory and Recording Research at the University of California, San Diego, and between 2018--2022, he was affiliated with the Institute of Advanced Studies, Technical University of Munich, where he held a four-year Hans Fischer Fellowship, funded by the German Excellence Initiative and the EU 7th Framework Program. Between August 2023 and January 2024, he was a Visiting Associate Professor at the School of Physical and Mathematical Sciences at Nanyang Technological University. Since 2024 he is a member of the Israel Young Academy. He is a recipient of several grants, including the ERC Consolidator Grant and the EIC Pathfinder Challenge. 
\end{IEEEbiographynophoto}

\begin{IEEEbiographynophoto}{Yiwei Zhang} received the B.A. and Ph.D. degrees in mathematics from Zhejiang University, Hangzhou, Zhejiang, China, in 2011 and 2016, respectively. From 2016 to 2017, he was a Post-Doctoral Researcher with the School of Mathematical Sciences, Capital Normal University, Beijing, China. From 2017 to 2019, he was a Post-Doctoral Researcher with the Department of Computer Science, Technion-Israel Institute of Technology, Haifa, Israel. He is currently a Professor with the School of Cyber Science and Technology, Shandong University, Qingdao, Shandong, China. He is also with State Key Laboratory of Cryptography and Digital Economy Security, the Key Laboratory of Cryptologic Technology and Information Security, Ministry of Education, Shandong University. His current research interests include coding theory and information security.
\end{IEEEbiographynophoto}

\end{document}